%% file: ex_article.tex
\documentclass{siamart190516}

\input{share}
\ifpdf
\hypersetup{
  pdftitle={Identifiability analysis of vaccination decision-making dynamics},
  pdfauthor={}
}
\fi

\begin{document}

\maketitle

\begin{abstract}
Variations in individuals'  perceptions of vaccination and  decision-making processes can give rise to poor vaccination coverage.
The future vaccination promotion programs
will benefit from  understanding this heterogeneity amongst groups within a population and, accordingly, tailoring the communication strategies.
Motivated by this, we developed a mechanistic model  consisting of a system of ordinary differential equations that
 categorizes individuals based on two factors:
 \emph{(i)}
perceived payoff gains for vaccination
and \emph{(ii)} decision-making strategies where we assumed that individuals may behave as either \emph{myopic rationalists}, going for a dose of vaccine if doing so maximizes their perceived  payoff gain, or \emph{success-based learners}, waiting to observe feedback on vaccination before deciding.
 We then investigated the global identifiability of group proportions  and  perceived payoff gains, that is, the possibility of globally retrieving these parameters by observing the error-free cumulative proportion of vaccinated individuals over time.
 To do so, for each group, we assumed a piecewise constant payoff gain and, for each time interval, obtained the so-called \emph{generalized input-output equation}.
We then proved the 
 global identifiability of these parameters under certain conditions.
 Global identifiability opens the door to  reliable estimations of the group proportions and their perceived payoffs.
\end{abstract}


\section{Introduction}
The results of past vaccination programs have revealed the heterogeneity of individuals' responses to  the opportunity of immunization.
Indeed, each individual perceives costs and benefits associated with vaccine uptake
due to their social environment, lived experiences  \cite{dolu2023covid}, gathered information  \cite{jones2012parents}, or their personal circumstances  \cite{goncu2021covid}, which make them  welcoming, strongly hesitant, or anywhere in between \cite{olson2020addressing}. 

These perceived benefits and costs manifest in  individuals' vaccine uptake behavior through their decision-making strategies.
More specifically, 
 some individuals might behave as \emph{myopic rationalists} or \emph{innovators} \cite{van2007new}, those who decide to get immunized if doing so maximizes their perceived instant payoff gain,  whereas others might behave as \emph{success-based learners} or \emph{imitators} \cite{van2015focus,pingle1995imitation}, those who rely on interactions with others, and, intuitively, decide to get vaccinated if their own perceived payoff for remaining unvaccinated fall short of the perceived payoff of vaccinated individuals.

In the context of COVID-19 vaccination program,
these perceived benefits and costs did not remain fixed over time; events such as
the emergence of delta variant or  vaccine mandate policies elicited different responses from society \cite{oliu2022effect}.
If public health authorities had known the population proportions of individuals with different perceptions  and decision-making strategies, 
they could have tailored their communication strategies more effectively to maximize vaccination coverage.
Yet, these proportions are not typically accessible to measurement.

In eyes of a modeler, 
this desire can be framed as developing a plausible model which captures the decision-making processes of a heterogeneous population where the proportions of groups appear as parameters, and then estimating  these parameters by  fitting the model to the data on vaccine uptake.
However, what if the model is fitted equally well with different proportions?

Thankfully, this question can be addressed by \emph{a priori} identifiability analysis of the developed model.
 Identifiability investigates whether it is at all possible to globally retrieve the parameters of a system with a known and error-free structure and noise-free inputs and outputs \cite{audoly2001global}.
In the context of vaccination, the vaccine supply is the input, and the vaccination progress is the measured output.

The theory and applications of  identifiability have attracted researchers over the years from different fields, such as automatic control \cite{glover1974parametrizations,denis2004equivalence, sarathchandra2024extending}, epidemiology \cite{sauer2021identifiability,dankwa2022structural}, and host infection \cite{xia2003identifiability,ciupe2022identifiability}.
Several approaches have been proposed to investigate identifiability, including but not limited to power series expansion \cite{pohjanpalo1978system}, local state isomorphism \cite{vajda1989state}, implicit function \cite{xia2003identifiability}, and differential algebra \cite{ljung1994global}, see \cite{chis2011structural,miao2011identifiability,anstett2020priori} for a comprehensive review.
Accordingly, several online tools and software packages have been dedicated to identifiability analysis of the constructed models \cite{bellu2007daisy,hong2019sian,julia }.
Differential algebra approach was introduced by \emph{Ritt} which involves in  constructing the so-called \emph{characteristic sets} of the  ideal generated by the differential equations describing the system \cite{ritt1950differential}.
The characteristic set then includes \emph{input-output relation}, that is, a set of normalized differential polynomials (set to zero) in terms of only known time-varying variables and unknown fixed parameters. 
The monomials of the input-output relations, under the so-called \emph{solveability}, have identifiable coefficients, see \cite{saccomani2003parameter}  for more details and \cite{hong2020global,ovchinnikov2022input,ovchinnikov2021computing} for consequences of violating the solvability condition.
The identifiability of the parameters then reduces to the one-to-oneness of the map from parameter spaces to the obtained coefficients.
The existence of the generalized input-output equations is guaranteed if
 the system dynamics and the output are rational in terms of the systems states and parameters.

In this paper, we investigate  identifiability of vaccination decision-making dynamics.
In this regard, in \cref{sec:problem-formulation} we provide a mechanistic model capturing the vaccine uptake in a well-mixed population.
The population is stratified into several groups in terms of  perceived payoff gain for vaccination and the decision-making strategies.
The developed model is a generalization of the one proposed in \cite{aghaeeyan2024revealing}, where the perceived payoff gain for vaccination was uniform across the population.
In \cite{aghaeeyan2024revealing}, the identifiability of the proportions of myopic rationalists and, consequently, imitators was shown.
The proportions in American and Canadian jurisdictions were then estimated by fitting the model to COVID-19 vaccine uptake data \cite{aghaeeyan2024revealing, aghaeeyan2024majority}.

In \cref{sec:main}, we then delve into the identifiability of the parameters appearing in the model including 
the proportion of each group.
The available results on identifiability  are not readily applicable to the developed model in \cref{sec:problem-formulation}; the perceived payoff for vaccine uptake may vary over time and get negative for some time interval resulting in change in the effective order of the system and, in turn, the input-output equation. 
In addition, we are interested in the identifiability of a general population consisting of $\n$ groups.
Hence, investigating  the 
 identifiability of the proportions of groups of individuals in vaccine uptake is more subtle and demands more through analysis.
 Our paper provides sufficient conditions for the identifiability of the parameters in the vaccination decision-making dynamics.
 \section{Problem Formulation} \label{sec:problem-formulation}
  We consider a well-mixed, fixed-size, and large population of individuals deciding whether to get vaccinated over the time interval $[0,T)$, where $T>0$.
       Once individuals get vaccinated, they remain vaccinated forever.    
   At each time $t$ individuals share an equal 
    dimensionless perceived payoff for vaccination $\pi(t)$, whereas the dimensionless perceived payoff  for remaining unvaccinated $\bar{\pi}(t)$ may vary across the population.
      In terms of decision-making strategies, each individual is either a \emph{(myopic) rationalist} or an \emph{imitator (also known as success-based learner)} \cite{lim2016experimental, van2015focus, aghaeeyan2024revealing}.
     At time $t$, an unvaccinated rationalist is a vaccine seeker  if only if her payoff for remaining unvaccinated $\bar{\pi}(t)$ falls short of the payoff for vaccination $\pi(t)$.
 An unvaccinated imitator, however, decides on vaccination through meeting vaccinated others and comparing her own current payoff  with those of others.

Individuals with the same decision-making strategies who have the same perceived payoff for remaining unvaccinated over the time $t \in [0,T)$ build up a group.
We assume that there are altogether $\n$ disjoint groups which are labeled by numbers $1,2,\ldots, \n.$
 The set of indices of rationalist (resp. imitative) groups is denoted by $\mathcal{R}$ (resp. $\mathcal{M}$).
The proportion of vaccinated individuals in group $i$, $i \in \{1,2,\ldots,\n\}$, at time $t$ is denoted by $\gamma_i(t)$, which satisfies $\gamma_i \in [0, \rho_i]$ where $\rho_i$ is the proportion of individuals in group $i$ satisfying $\sum_{i=1}^\n \rho_i = 1$.
  
 The proportion of vaccine seekers in rationalist group $i \in \mathcal{R}$ at time $t$, denoted by $s_i(t)$, equals the total unvaccinated proportion $\rho_i - \gamma_i(t)$  if their excess payoff  $\Delta \pi_i(t) = \pi(t) - \bar{\pi}_i(t)$ is positive, where  $\bar{\pi}_i(t)$ is the perceived payoff of group $i$ for remaining unvaccinated at time $t$, and zero otherwise.
For imitative group $i \in \mathcal{M}$, the proportion of vaccine seekers $s_i(t)$ equals the proportion of unvaccinated individuals  who imitate the vaccinated individuals.
This imitation 
is modeled by the probability of meeting vaccinated individuals,
$\sum_{i=1}^{\n} \gamma_i(t)$,
times the probability of imitation, which is given by the pairwise comparison term $\Delta \pi_i(t)$ if it is positive, and zero otherwise.
Thus, the proportion of vaccine seekers in group $i$ at time $t$ equals
\begin{equation} \label{eqn:vax-seeker}
     s_i(t) = 
    \begin{cases}
     \big(\rho_i - \gamma_i(t)\big) \1\big(\Delta \pi_i(t)\big),\quad &\text{ if } i \in \mathcal{R},\\
    \big(\rho_i - \gamma_i(t)\big) \big(\sum_{i=k}^{\n} \gamma_k(t)\big)  \Delta{\pi}_i(t) \1\big(\Delta \pi_i(t)\big), &\text{ if } i \in \mathcal{M},
   \end{cases}
\end{equation}
where
 $\1(x)$ is one for $x>0$ and zero otherwise.
We drop the word ``excess.''

Whether vaccine seekers can receive a dose of vaccine depends on the vaccine supply denoted by $u(t)$.
If demand falls short of supply, i.e., $\sum_{j=1}^\n s_j(t) < u(t)$, all vaccine seekers can get vaccinated with the maximum  rate of vaccination denoted by $\kappa$.
If the demand exceeds the supply, we assume that the available supply is distributed uniformly and randomly among the vaccine seekers and, in turn, the rate of vaccination reduces to $\kappa \frac{u(t)}{\sum_{j=1}^\n s_j(t)}$.
The  proportion of vaccinated individuals in group $i$, $\gamma_i$, will change with the following rate:
\begin{align} \label{eqn:mean-dynamics}
    \dot{\gamma}_i(t) =  \kappa s_i(t) \min \{1, \frac{u(t)}{\sum_{j=1}^\n s_j(t)} \}.
 \end{align}
 Equivalently,  if group $i$ is a rationalist, we have
  \begin{align}
   \label{eq:dot-adopters-br}
    &\dot{\gamma}_i(t) = \kappa \big(\rho_i -\gamma_i(t)\big)\1\big(\Delta  \pi_i(t_l)\big)
      \\
 &  \times\! \min \{1, \frac{u(t)}{\sum\limits_{\!j \in \mathcal{M}}\! \big(\rho_j -\gamma_j(t)\big)\big(\!\sum\limits_{k=1}^{\n}\! \gamma_k(t)\big)\Delta \pi_j(t)\1\big(\Delta \pi_j(t)\big) \!+ \!\sum\limits_{\!j \in \mathcal{R}}\!\big(\rho_j -\gamma_j(t)\big) \1\big(\Delta \pi_j(t)\big)}  \} \nonumber
 \end{align}
 and  if group $i$ is an imitator, we have
\begin{align} \label{eq:dot-adopters-imitator}
    &\dot{\gamma}_i(t)  = \kappa \big(\rho_i -\gamma_i(t)\big) \big(\sum_{k=1}^{\n} \!\gamma_k(t)\big)\Delta \pi_i(t)\1\big(\Delta \pi_i(t)\big)
   \\
 & \times \min \{1, \frac{u(t)}{\sum\limits_{\!j \in \mathcal{M}} \!\big(\rho_j -\gamma_j(t)\big)\big(\sum\limits_{k=1}^{\n} \gamma_k(t)\big)\Delta \pi_j(t)\1\big(\Delta \pi_j(t)\big) \!+ \!\sum\limits_{\!j \in \mathcal{R}}\!\big(\rho_j -\gamma_j(t)\big) \1\big(\Delta \pi_j(t)\big)}  \}. \nonumber
 \end{align}
 \begin{remark}
Individuals' perceived payoffs may  change in response to some impactful events.
In the COVID-19 vaccine campaign,
  events such as the announcement of vaccine certificates \cite{oliu2022effect} or reports of severe vaccine-related side effects \cite{zhang2023impact} may have impacted the individuals' payoffs.
 It is well-known that clinician recommendations improve attitudes toward HPV vaccination \cite{vollrath2018meaningful}.
For flu vaccination, public immunization programs increased vaccine uptake in Ontario, Canada \cite{flushot}, which could serve as a  proxy for changes in individuals' perceived payoffs.
A piecewise constant function seems to effectively approximate the  changes in payoff values in response to such events and interventions, the timing of which can often be identified through gray literature.
Accordingly, we assume that these time instants are known \emph{a priori}.
\end{remark}
\begin{assumption}\label{ass:subintervals-vaccine-surplus-period}
There is a \emph{a priori} known set $\mathcal{T}$ which partitions $[0,T)$ by $\s$ disjoint non-empty subintervals $[T_{j-1}, T_{j})$, $j = 1,2,\ldots, \s,$ where  $T_0 = 0$, $T_{\s} = T$,  $T_{j-1} < T_j$, and 
   for every $[T_{j-1},T_j) \in \mathcal{T}$, the  payoff  values of all $\n$ groups are constant.
 \end{assumption}
\begin{remark} \label{remark:labeling}
    The labeling of the groups could be based on some factors such as gender, socioeconomic status, demographic profiles, affiliations with organizations, etc.
    As will be shown later, this labeling allows us to leverage relevant information about the groups in each subinterval to rank each group's payoff.
\end{remark}
 Motivated by COVID-19 vaccine distribution, where there was  shortage in supply for the first few months of vaccine roll-out, we make the following assumption.
\begin{assumption} \label{ass:two-subintervals-shortage-surplus}
There exists a time instant $t^* \in [0,T_1)$  such that for all 
$t$ in the interval $[0,t^*)$ 
the demand exceeds the vaccine supply, which is integrable.
In addition, the demand falls short of the vaccine supply for all $t\geq t^*$. 
\end{assumption}
 In view of \Cref{ass:two-subintervals-shortage-surplus}, \emph{(i)} the time interval $[0,T)$ is split into two parts: 
 period of limited vaccine supply $[0,t^*)$ and vaccine surplus period $[t^*,T)$ and  \emph{(ii)} the payoff values of the groups remain constant during the period of limited vaccine supply and for some time afterward.
 We  refer to the time instant $t^*$ as ``vaccine abundance time.''
\begin{remark}
There might be some 
   groups whose payoff values remain non-positive over the entire time interval $[0,T)$, disregard of their decision-making strategies. 
   We refer to these groups as \emph{vaccine refusing} groups.
A \emph{non-vaccine refusing} group is then a group with a positive payoff value during at least one subinterval.
\end{remark}
\begin{example} \label{example:first}
 Consider a  population of individuals stratified into two age-based groups: seniors and non-seniors.
A 40-day vaccine campaign begins on day zero, during which all individuals have the opportunity to receive a vaccine dose, with sufficient supply available from the start.
Regarding vaccination behavior, each age-based group is  stratified into three subgroups: vaccine refusers, imitators, and rationalists.
Group 1 (resp. Group 2) consists of rationalist seniors (resp. non-seniors), group 3 (resp. group 4) of imitative seniors (resp. non-seniors), and group 5 (resp. group 6) of vaccine-refusing seniors (resp. non-seniors).
The dynamics read as 
\begin{equation}
    \begin{aligned}
        \dot{\gamma}_i(t) &= \kappa \big(\rho_i - \gamma_i(t)\big) \1\big(\Delta \pi_i(t)\big), \quad i =1,2, \\
        \dot{\gamma}_i(t) &= \kappa \big(\rho_i - \gamma_i(t)\big) \big(\sum_{k=1}^6 \gamma_k(t) \big) \Delta \pi_i(t)\1\big(\Delta \pi_i(t)\big), \quad i =3,4, \\
        \dot{\gamma}_i(t) &= 0, \quad i=5,6,
    \end{aligned}
\end{equation}
where $\gamma_i(0) = 0$, for $i = 1,2,\ldots,6.$
Individuals' perceived excess payoffs for vaccination remain constant until day $20$ at which public health   offers some incentive for vaccination.
Neither the proportions of the groups nor their payoff values are known, and
the available temporal data  is the cumulative proportion of vaccinated individuals, \Cref{fig:example-introduction}.
Revealing the proportions of groups and the evolution of the payoff values over time helps to better explain the overall vaccination progress. 
\end{example}
\begin{figure} 
    \centering
    \includegraphics[width=0.95\linewidth]{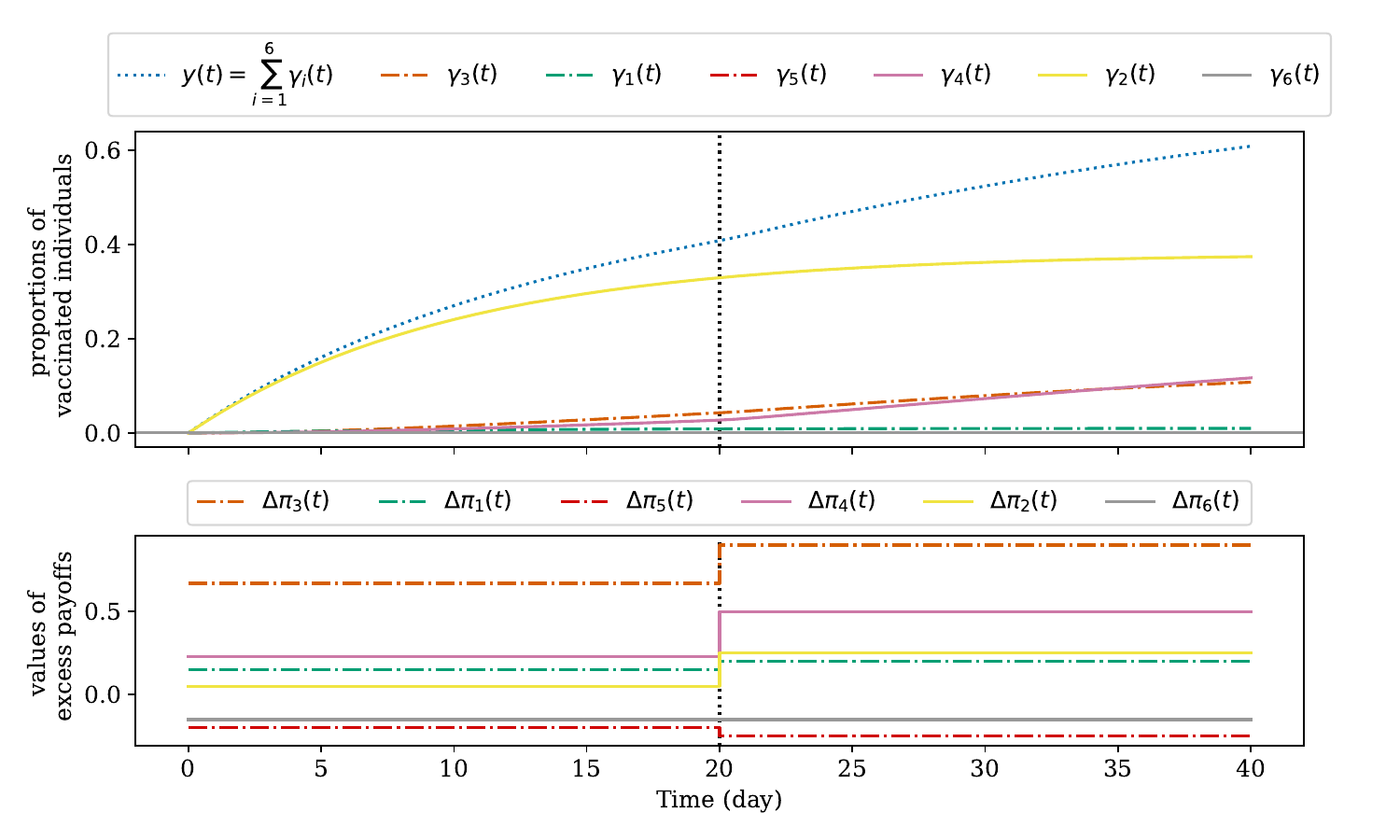}
    \caption{Vaccination progress in \Cref{example:first}. 
    Upper panel: the available data is the cumulative  proportion of vaccinated individuals $y(t)$ (dotted line).
    The proportions of vaccinated rationalist seniors, imitative seniors, and vaccine-refusing seniors are denoted by $\gamma_1(t)$, $\gamma_3(t)$, and $\gamma_5(t)$, respectively (dash-dotted lines).
    For non-seniors, the corresponding proportions are denoted by $\gamma_2(t)$, $\gamma_4(t)$, and $\gamma_6(t)$, respectively (solid lines).
    Lower panel:
    the perceived  payoff for each group is denoted using the same subscript as that of the ratio of vaccinated individuals.
    The payoffs change at day $20$ in response to the announcement of new incentives.
    }\label{fig:example-introduction}
    \end{figure}
Motivated by the above example, we are interested in the \emph{global} identifiability of the parameters appearing in  \eqref{eq:dot-adopters-br} and
\eqref{eq:dot-adopters-imitator}, for $i = 1,2,\ldots,\n$, given the opportunity to record the cumulative  proportion of vaccinated individuals over time.
In the following, we provide the definition of \textit{global} identifiability in the sense of \cite{audoly2001global}.
\begin{definition}[Identifiability] \label{def:identifiability}
    Consider the system
\begin{equation} \label{eqn:complete-system}
    \Sigma^{\bm \theta} = 
    \left\{
    \begin{aligned}
        & \dot{\bm \gamma}(t) = \bm f\big({t},\bm \gamma(t), \bm u(t), \bm \theta \big), \\
        & \bm \gamma(0) = \bm \gamma_0(\bm \theta),\\
        & \bm y(t, \bm \theta) = \bm g \big(\bm \gamma(t), \bm \theta \big),
        & \bm d(\bm \theta) = \bm 0,
    \end{aligned}
    \right.
\end{equation}
    where $\bm \gamma$ is the $\mathtt{x}$-dimensional state variable vector, 
    $\bm \theta \in \bm \Theta $ is the ${\theta}$-dimensional parameter vector,
    $\bm u$ is the 
    $\mathtt{u}$-dimensional  input vector, $\bm y$ is the $\mathtt{y}$-dimensional output vector,  $\bm \Theta$ is the parameter space,
     $\bm f: \mathbb{R} \times \mathbb{R}^{\mathtt{x}} \times \mathbb{R}^{\mathtt{u}} \times  \bm \Theta  \to \mathbb{R}^{\mathtt{x}}$, $\bm g:  \mathbb{R}^{\mathtt{x}} \times  \bm \Theta  \to \mathbb{R}^{\mathtt{y}}$, and $\bm d$ is the $\mathtt{d}$-dimensional vector of equality constraints.
    The parameter $\theta_i$, $i \in \{1,\ldots, \mathtt{\theta} \}$, is  \emph{globally identifiable}
    if and only if for almost any $\bm \theta' \in \bm \Theta$, we have
    \begin{equation} \label{eqn:identifiability}
      \big(\forall t\geq 0\ {y}(t, \bm \theta) = {y}(t, \bm \theta') \big) \implies \theta_i = \theta'_i.
      \end{equation}
\end{definition}
\begin{remark}
   In reference \cite{audoly2001global}, an equivalent definition is provided for \emph{a priori} global identifiability, where the term ``a priori'' emphasizes the  assumptions of an error-free model and perfect and noise-free data.
   Here, we also consider these  assumptions without using the term ``a priori.''
   Similar definitions are also given for \emph{structural} identifiability, where the term ``structural'' indicates that the identifiability results hold for almost any parameter vector in the parameter space \cite{walter1981unidentifiable}. 
   \Cref{def:identifiability} also captures this characteristic.
  However, as definitions of structural identifiability often disregard initial conditions \cite{eisenberg2013identifiability}, we avoid using the term.
\end{remark}
In the context of vaccine uptake 
dynamics governed by equations \eqref{eq:dot-adopters-br} and  \eqref{eq:dot-adopters-imitator}, for $i = 1,2,\ldots, \n,$ the values of $\mathtt{x}$, $\mathtt{y}$, and $\mathtt{u}$ are respectively $\n$, $1$, and $1$. 
 \Cref{ass:subintervals-vaccine-surplus-period} imposes a piece-wise constant structure on the  payoff function of each group $i$, and accordingly the identifiability of the  payoffs boils down to that of the  vectors $\bm \Delta \bm \pi_i$, $i =1,\ldots, \n$, where
  $\bm \Delta  \bm \pi_i = \big(\Delta \pi_i(T_0),\ldots,\Delta \pi_{i}(T_{\s-1}) \big)$.
Parameter $\kappa$ and the proportions $\rho_i$, $i = 1,2,\ldots, \n$, are also unknown. 
\begin{definition}[System]
    We define the following dynamical system
    \begin{equation} \label{eqn:complete-system1}
        \Sigma^{\bm \theta} = 
            \left\{
            \begin{aligned}
                & \dot{\bm \gamma}(t) = \bm f(t,\bm \gamma(t), u(t), \bm \theta), \quad \bm \gamma(0) = \bm 0, \\
                & f_i(t,\bm \gamma(t), u(t), \bm \theta) = 
                \begin{cases}
             \eqref{eq:dot-adopters-br},       & i \in \mathcal{R}, \\
            \eqref{eq:dot-adopters-imitator}, &i \in \mathcal{M},  
                \end{cases}
                \\
                &{y}(t, \bm \theta)= \sum_{i=1}^\n \gamma_i(t), \\
                &  \sum_{i=1}^\n \rho_i = 1,
    \end{aligned}
    \right.
    \end{equation}
    where 
        $f_i$ is the $i^{\text{th}}$ component of $\bm f,$
      $\bm \theta = (\kappa,  \rho_1,\ldots, \rho_{\n}, \bm \Delta  \bm \pi_1, \bm \Delta \bm \pi_2, \ldots, \bm \Delta  \bm \pi_{\n})$ 
      and 
      $\bm \Delta  \bm \pi_i = \big(\Delta \pi_i(T_0),\ldots,\Delta \pi_{i}(T_{\s-1}) \big)$, for $i \in \{1,\ldots, \n\}$.
\end{definition}
In the next section, we present our results on the identifiability of the vector of parameters $\bm \theta$. 
From now on, we will omit the dependency of the output on the parameters and use the notation $y(t)$.
\section{Results}
\label{sec:main}
Let $x_i(t)$ denote  
  the proportion of unvaccinated individuals in group $i$ at time $t$, i.e.,
    $x_i(t)  = \rho_i - \gamma_i(t),$
the output $y(t)$ can then be rewritten as follows:
 \begin{equation}\label{eq:y}
    y(t) = 1 -\sum_{i=1}^{\n} x_i(t).
\end{equation}
In view of \cref{ass:subintervals-vaccine-surplus-period} and \cref{ass:two-subintervals-shortage-surplus},
 the rate of change of the proportion of unvaccinated individuals in group $i$ 
 over $t\in [T_l,T_{l+1}) \in \mathcal{T}$, for $l \in \{0,1,\ldots, \s-1\}$, is governed by 
 \begin{flalign}\label{eq:dotGamma}
    \dot{x}_i(t) = &&
 \end{flalign}
 \vspace{-20pt}
 \begin{equation*} 
 \scalebox{0.82}{$
  \begin{cases}
         \!-\kappa\! x_i(t)\!\1\big(\!\Delta \pi_i(T_l)\big)\! \min \{1, \frac{u(t)}{\!\displaystyle\sum_{i \in \mathcal{M}} x_i(t)y(t)\Delta \pi_i(T_l)\1\big(\Delta \pi_i(T_l)\big)\! + \!\displaystyle\sum_{i \in \mathcal{R}}x_i(t) \1\big(\Delta \pi_i(T_l)\big)}  \},
        & i \in \mathcal{R}, \\
         \!-\kappa \!x_i(t)\! y(t)\Delta \pi_i(T_l)\1\big(\Delta \pi_i(T_l)\big)\! \min \{\!1,\! \frac{u(t)}{\!\displaystyle\sum_{i \in \mathcal{M}} x_i(t)y(t)\Delta \pi_i(T_l)\1\big(\Delta \pi_i(T_l)\big)\! + \!\displaystyle\sum_{i \in \mathcal{R}}x_i(t) \1\big(\Delta \pi_i(T_l)\big)}  \}, & i \in \mathcal{M},
     \end{cases}
     $}
 \end{equation*}
 and $x_i(0) = \rho_i$ for $i \in \{1,\ldots,\n\}$.
 It can be shown that the systems governed by \eqref{eq:y}, \eqref{eq:dotGamma} is \emph{input-output-equivalent} to the system defined in \eqref{eqn:complete-system1}, that is, both systems generate the same output trajectory given the same parameter vector and input function.
 Input-output equivalent systems have the same identifiable and unidentifiable parameters \cite{eisenberg2013input}.
Hence,  we investigate the identifiability of the vector of parameters $\bm \theta$ defined in \eqref{eqn:complete-system} through analyzing the system governed by \eqref{eq:y} and \eqref{eq:dotGamma}.

We proceed with analyzing the identifiability of parameter $\kappa$ and the payoff values of the groups.
 We then investigate the identifiability of the proportions of the groups.
The proof of the results are provided in \Cref{sec:proof}.
\begin{proposition} \label{lem:identifiability-kappa}
  The maximum rate of vaccination $\kappa$ is globally identifiable if  \Cref{ass:subintervals-vaccine-surplus-period} and \Cref{ass:two-subintervals-shortage-surplus} hold and either there exists at least one non-vaccine refusing rationalist group or the vaccine abundance time $t^*$ is greater than zero and the vaccine supply $u(t)$ is not identically zero.
\end{proposition}
 \subsection{Identifiability analysis of the payoff functions} \label{subsec:identifiability-payoff-steady-state}
  In what follows, we investigate the identifiability of the 
  payoff values of the groups in each subinterval, separately.
  
 Depending on the signs of the payoff values, a subinterval $[T_{l},T_{l+1}) \in \mathcal{T}$ falls into one of the following  distinct categories:
\emph{(i)} $\Delta \pi_i(T_{l}) \leq 0$ for all $i \in \{1,\ldots,\n\}$, 
\emph{(ii)} $\Delta \pi_i(T_{l}) \leq 0$ for all $i \in \mathcal{M}$ and $\Delta \pi_i(T_{l}) > 0$  for some $i \in \mathcal{R}$,
\emph{(iii)} $\Delta \pi_i(T_{l}) \leq 0$ for all $i \in \mathcal{R}$ and $\Delta \pi_i(T_{l}) > 0$  for some $i \in \mathcal{M}$, and
\emph{(iv)} none of the previous cases.
 
For the first case, we have the following results: 
 \begin{lemma} \label{lem:unidentifiability-of-negative-payoffs}
      The non-positive components of the vector of  payoff values $\bm \Delta \bm \pi_i$, for $i \in \{1,\ldots,\n\}$, are unidentifiable.
  \end{lemma}
  
In the second case, \Cref{lem:unidentifiability-of-negative-payoffs} yields the unidentifiability of the  payoff values of the imitators and those rationalists with non-positive  payoff values.
As for positive  payoff values of rationalists, we have the following proposition:
 \begin{proposition} \label{prop:unidentifiability-of-rationalist-payoff}
     Let $[T_{l}, T_{l+1}) \in \mathcal{T}$ be a time interval   
    during which  $\Delta \pi_i(T_{l}) \leq 0$ for all $i \in \mathcal{M}$ and $\Delta \pi_i(T_{l}) > 0$  for some $i \in \mathcal{R}$.
    Then the $(l +1)^{\text{th}}$ component of $\bm \Delta \bm \pi_i$ for all $i\in\{1,2,\ldots,\n\}$ 
    is unidentifiable.
 \end{proposition}

 \begin{corollary} \label{cor:unidentifiability-of-payoff-rationalists}
     Let  $i \in \mathcal{R}$.
     Then, none of the components of the vector $\bm \Delta \bm \pi_i$ is identifiable.
 \end{corollary}
 
In the following subsections, we investigate the identifiability of the  payoff values  for a time interval where at least one imitative group has positive payoff value.
In this regard, we adopt a common approach to identifiability  which involves 
deriving the so-called \emph{(generalized) input-output equation}, that is, a normalized differential polynomial equation consisting of monomials in terms of the input, output and their (higher-order) derivatives, and the parameters \cite{eisenberg2013input}.
By ``monomial,'' we refer to a product of powers of output and its derivatives with nonnegative integer exponents. 
Based on \Cref{thm:Eisenberg}, the coefficients of monomials in a normalized input-output equation are globally identifiable--under solveability condition.
 Hence, if
 there is a one-to-one map between the coefficients of the  monomials of the  input-output equation and the payoff values,  the identifiability of the payoff values is immediate. 
 
 In the subsections
 \ref{subsec:identifiability-payoff-steady-state-imitating-only},
 \ref{subsec:identifiability-payoff-steady-state-imitating-and-br}, and
 \ref{sub:Identifiability analysis of the proportions of unvaccinated individuals at},
 we only focus on vaccine surplus period $[t^*,T)$.
For the ease of notation, we define 
 \begin{equation}
    t_j =
    \begin{cases}
        T_j & \text{ if } j \in \Z_{>0},\\
        t^* & \text{ if } j = 0.\\
    \end{cases}
\end{equation}
 Denote by $\bar{x}(t_l)$ the  proportion of unvaccinated individuals at time instant $t_l$ who have non-positive  payoff values during time interval $[t_l,t_{l+1})$, i.e.,
\begin{equation} \label{eq:barx}
    \begin{aligned} 
       \bar{x}(t_l) &=  \sum_{\substack{{j: j \in \{1,2,\ldots,\n\},}\\ {\Delta \pi_j(t_l)\leq 0}}} x_j(t_l).
    \end{aligned}
\end{equation}
Let $\mathcal{P}(t_l)$ denote the set of positive payoff values of imitators during  $[t_l, t_{l+1})$, i.e.,
\begin{equation} \label{eq:set-of-payoff}
\mathcal{P}(t_l) = \bigcup_{\substack{{i: i \in \mathcal{M},}\\ {\Delta \pi_i(t_l)>0}}} \{\Delta \pi_i(t_l)\}.
\end{equation}
The cardinality of the set $\mathcal{P}(t_l)$ is denoted by $n_l$.
Indeed, $n_l$ is the number of  distinct positive  payoff values of imitators during the subinterval $[t_l, t_{l+1})$.
Order the members of the set $\mathcal{P}(t_l)$ in an arbitrary order as a vector $\hat{\bm \pi}(t_l)$ of dimension $n_l$.
Define
\begin{equation} \label{eq:xhat-for-vandermonde}
    \hat{x}_i(t) = \sum_{\substack{{j: j \in \mathcal{M},}\\ {\Delta \pi_j(t_l) =  \hat{\pi}_i(t_l)}}} x_j(t), \quad  i=1,2,\ldots,n_l,
\end{equation}
where $\hat{\pi}_i(t_l)$, $i \in \{1,\ldots,n_l\}$, denotes the $i^{\text{th}}$ component of the vector $ \hat{\bm\pi}(t_l).$
The variable $\hat{x}_i(t)$ equals the  proportion of unvaccinated imitators having the payoff value $\hat{\pi}_i(t_l)$ during the time interval $[t_l,t_{l+1}).$
\subsubsection{A subinterval in which no rationalists but at least one imitative group gets vaccinated} \label{subsec:identifiability-payoff-steady-state-imitating-only}

Assume that during some subinterval $[t_l, t_{l+1}),$ $l \in \{0,1,\ldots,\s-1\},$ no rationalist groups have positive  payoff values, but
at least one imitative group has. 
The output \eqref{eq:y} then satisfies 
\begin{equation}\label{eq:the-output-for-vandermonde}
    \begin{aligned}
  y(t) &= 1-  \sum_{i=1}^{\n} {x}_i(t)  & \\
  & =  1-  \sum_{i=1}^{n_l} \hat{x}_i(t) - \bar{x}(t_l).
    \end{aligned}
\end{equation}
From \eqref{eq:dotGamma} and \eqref{eq:xhat-for-vandermonde}, the derivative of $\hat{x}_i$ for $i = 1,\ldots,n_l$ and $t \in [t_{l},t_{l+1})$ equals
$
    \dot{\hat{x}}_i(t) = -\kappa \hat{x}_i(t) y(t) \hat \pi_i(t_l).
$
Moreover, the value of $\bar{x}(t_l)$ is  constant over the subinterval $[t_l, t_{l+1})$.
Hence, by differentiating \eqref{eq:the-output-for-vandermonde}, the derivative of the output reads
$\dot{y}(t) 
    = \kappa y(t) \sum_{i=1}^{n_l}\hat \pi_i(t_l)\hat{x}_i(t).$
By defining
 $f_{(j,k)}(t)$ for $k,j \in \Z$ as
    \begin{equation}
        \begin{cases} \label{eq:f-j-k}
    f_{(j,k)}(t)  =\dot{f}_{(j,k-1)}(t) - \kappa y(t) f_{(j-1,k-1)}(t), & \\
    f_{(1,1)}(t) = \kappa y(t), & \\
    f_{(j,k)}(t) = 0 \quad \text{if } j > k \text{ or }  j \leq 0 \text{ or } k\leq 0, 
        \end{cases}
    \end{equation}    
  the $k^{\text{th}}$ order derivative of the output $y$, denoted by $y^{(k)}$, can be written in the following compact form.
\begin{lemma} \label{lem:induction-For-Dot-y}
   For a time interval $[t_l,t_{l+1})$, $l \in \{0,1,\ldots,\s-1\}$, where the payoff values of rationalist groups are not positive, but those of some imitative groups are,  we have
    $
    y^{(k)}(t) = \sum_{j=1}^k f_{(j,k)}(t) \sum_{i=1}^{n_l}  \hat \pi_i^j(t_l) \hat{x}_i(t), \quad  k \in \Z_{>0}.
    $
\end{lemma}
We are looking for the input-output relation, which can be found if we could write the term $\sum_{i=1}^{n_l}  \hat \pi_i^{k}(t_l)\hat{x}_i$, for some $k$, 
in terms of the output and its (higher-order) derivative.
In view of \eqref{eq:f-j-k}, the function
$f_{(j,k)}(t)$ is described by the output and its derivatives. 
Hence,
based on \Cref{lem:induction-For-Dot-y}, the term
$\sum_{i=1}^n  \hat \pi_i^{k}(t_l)\hat{x}_i$ for $k >1$
 can be recursively written in terms of  $\kappa$, the output and its  derivatives, and $\sum_{i=1}^j  \hat \pi_i^{j}(t_l)\hat{x}_i$ where $j<k$.
More specifically, let
\begin{equation} \label{eq:alpha-series}
    \begin{cases}
    \displaystyle
    \alpha_k(t) = \frac{y^{(k)}(t) - \sum_{j=1}^{k-1} f_{(j,k)}(t) \alpha_j(t)} {(-1)^{k+1}\kappa^k y^k(t)},  & \\
    \displaystyle
    \alpha_1(t) = \frac{\dot{y}(t)}{\kappa y(t)},\\
        \displaystyle
    \alpha_0(t) = 1-y(t), &
    \end{cases}
\end{equation}
then for $k \in \Z_{> 0}$ and $t \in [t_l, t_{l+1})$, $l \in \{0,1,\ldots,\s-1\}$, where no rationalist groups have positive payoff values, we have
$
  \sum_{i=1}^{n_l}  \hat \pi_i^{k}(t_l)\hat{x}_i(t)    =  \alpha_k(t).
$
Note that the output $y(t)$ is strictly greater than zero for $t \geq t^* + \epsilon$, for an arbitrarily small positive $\epsilon$, and, in turn, the function $\alpha_k(t)$, $k\in  \mathbb{Z}_{>0}$, is bounded.
In view of \eqref{eq:f-j-k} and \eqref{eq:alpha-series}, the function $\alpha_k(t)$ is written based on the output and its  derivatives and parameter $\kappa$.
We can then construct the matrix equation 
\begin{equation} \label{eq:V-gamma-alpha}
\matr{V} \hat{\bm x}(t) = \bm \alpha(t),
\end{equation}
where $\bm \alpha(t) = \big(\alpha_0(t), \alpha_1(t), \ldots, \alpha_{n_l}(t)\big)^\top$, 
$\hat{\bm x}(t)= \big(\hat{x}_1(t),\hat{x}_2(t), \ldots, \hat{x}_{n_l}(t), \bar{x}(t_l)\big)^\top$, and
\begin{equation} \label{eq:vandermonde}
\matr{V}  =
    \begin{pmatrix}
       1 & 1  & \cdots & 1 & 1  \\
       \hat \pi_1(t_l) & \hat \pi_2(t_l) & \cdots & \hat \pi_{n_l}(t_l) & 0\\
       \vdots & \vdots & \ddots & \vdots & \vdots \\
        \hat \pi_1^{n_l-1}(t_l) &  \hat \pi_2^{n_l-1}(t_l)  & \cdots &  \hat \pi_{n_l}^{n_l-1}(t_l)   & 0\\[0.2cm]
        \hat \pi_1^{ n_l}(t_l) &  \hat \pi_2^{n_l}(t_l)  & \cdots &  \hat \pi_{n_l}^{n_l}(t_l)   & 0
    \end{pmatrix}.
\end{equation}
Matrix $\matr{V}$ is indeed a Vandermonde matrix of size $\big(n_l+1)\times(n_l+1\big)$. 
It is known that provided  $\hat \pi_i(t_l) \neq \hat \pi_j(t_l)$ for $i \neq j$, matrix $\matr{V}$ is invertible \cite{man2017computing}.
 From \eqref{eq:V-gamma-alpha} we have
\begin{equation} \label{eq:Gamma-based-on-Vinv}
\hat{\bm x}(t) = \matr{V}^{-1} \bm \alpha(t)
\end{equation}
which holds for all $[t_l, t_{l+1})$ during which no rationalist groups have positive payoff values.
The above relation is a set of $\big(n_l+1\big)$ equations where their
right-hand sides are in terms of the functions $\alpha_j(t)$, $j=0,1,\ldots,n_l$, and the parameters $\hat \pi_i(t_l)$, $i=1,2,\ldots,n_l$,  and consequently free of the state variables $\hat{x}_i(t)$.
The last component of vector $\hat{\bm x}(t)$, $\Bar{x}(t_l)$, is the proportion of unvaccinated individuals who have non-positive  payoff values for $t \in [t_l, t_{l+1})$, which is constant.
More specifically,
during the time interval $[t_l, t_{l+1})$, the term $\Bar{x}(t_l)$ is 
indeed a parameter.
Hence,
 the last equation in \eqref{eq:Gamma-based-on-Vinv} is the 
 input-output relation.
Analyzing this input-output relation yields the identifiability of the set $\mathcal{P}(t_l)$.
\begin{proposition} \label{lem:case1}
       For a subinterval $[t_l,t_{l+1}),$ $l \in \{0,1,\ldots,\s-1\}$, where the payoff values of rationalist groups are not positive, but those of some  imitative groups 
    are,
   the  set $\mathcal{P}(t_l)$
    is globally identifiable if Assumptions \ref{ass:subintervals-vaccine-surplus-period}, and \ref{ass:two-subintervals-shortage-surplus} hold.
\end{proposition}
\Cref{lem:case1} shows the  identifiability of the set of positive payoff values of imitative groups $\mathcal{P}_l$.
However, if we consider only one subinterval, it is not possible to associate each member of the set with a specific imitative group unless, based on some  contextual factor (see \Cref{remark:labeling}), there is a known ordering among them; that is for each $i,j \in \{k \in \mathcal{M} \vert \Delta \pi_k(t_l) > 0\}$, the value of $\max \{\Delta \pi_i(t_l), \Delta \pi_j(t_l\}$ is known.
\begin{corollary} \label{cor:identifiability-of-the-cis-imitators}
   Under the conditions of \Cref{lem:case1},
   the $(l+1)^{\text{th}}$ component of the vector $\bm \Delta \bm \pi_i$, for $i\in \mathcal{M}$, is
      identifiable if it is positive and the ordering of the positive payoff values of imitators is known.
\end{corollary}
\begin{corollary} \label{cor:identifiability-xbar-imitators}
    Under the conditions of \Cref{lem:case1}, the proportion of unvaccinated individuals with non-positive  payoff values for the time interval $(t_l, t_{l+1})$, $l \in \{0,1,\ldots,\s-1\}$, $\bar{x}(t_l)$, is globally identifiable.
\end{corollary}
\subsubsection{A subinterval in which both rationalist and imitative groups get vaccinated} \label{subsec:identifiability-payoff-steady-state-imitating-and-br}
Here,
we take a similar approach to the one we took in the preceding section, i.e, we derive the  input-output equation.
Consider the time interval $[t_l,t_{l+1})$, $l \in \{0,1,\ldots,\s-1\}$, during which at least one rationalist and one imitative groups have positive payoff values.
In view of \eqref{eq:dotGamma}, for rationalist groups 
we  have 
$
\dot{x}_i(t) = -\kappa x_i(t) \quad \text{for } i \in \mathcal{R} \text{  s.t. } \Delta \pi_i(t_l)>0.
$
Let us define $x'$ as the  proportion of unvaccinated rationalists with positive payoff value during this subinterval
\begin{equation} \label{eq:xprime}
 {x}'(t) = \sum_{\substack{{i: i \in \mathcal{R},}\\ {\Delta \pi_i(t_l)>0}}}  x_i(t).
 \end{equation}
 We then have
$
\dot{x}'(t) = -\kappa x'(t).
$
It can be shown that the $k^{\text{th}}$ order derivative of the output for $k \in \mathbb{Z}_{>0}$ will be 
\begin{equation}\label{eq:y-and-its-derivatives-R-and-M}
    y^{(k)}(t) = \sum_{j=1}^k f_{(j,k)}(t) \sum_{i \in \mathcal{M} } \Delta \pi_i^j(t_l) \1\big( \Delta \pi_i(t_l)\big) x_i(t) + (-1)^{k-1} \kappa^k x'(t),
\end{equation}
The input-output equation can be obtained similar to the previous case, just that
here besides function $\alpha_k$, we need to define an additional recursive function $\beta_k$.
 If $t^* = 0$ or the vaccine supply is identically zero over $[0,t^*)$, we consider $[\epsilon, t_1)$ for $l=0$ and an arbitrarily small positive $\epsilon$.
\begin{lemma} \label{lem:ci-gammai-based-on-alphas-betas}
Under Assumptions \ref{ass:subintervals-vaccine-surplus-period} and \ref{ass:two-subintervals-shortage-surplus},
    for a subinterval $[t_{l}, t_{l+1})$, $l \in \{0,1,\ldots,\s-1\}$, where at least one rationalist group and one imitative group have positive payoff values and for $k \in \mathbb{Z}_{>0}$
    we have
    $\sum_{i \in \mathcal{M}} \Delta \pi_i^k(t_l) \1 \big(\pi_i(t_l)\big)x_i(t) = \alpha_k(t) - x'(t) \beta_k(t),$
    where
    \begin{equation}
    \label{eq:betas}
    \begin{cases}
    \displaystyle
        \beta_k(t) = \frac{-\sum_{j=1}^{k-1} f_{(j,k)(t)}\beta_j(t) + (-1)^{k-1} \kappa^k}{(-1)^{k+1}\kappa^k y^k(t)}, &\\
        \displaystyle
        \beta_1(t) = 1/y(t),&
         \\
        \displaystyle
        \beta_0(t) = 1,&
    \end{cases}
\end{equation}
    and $x'(t)$ and $\alpha_k(t)$ are defined in \eqref{eq:xprime} and \eqref{eq:alpha-series}, respectively.
\end{lemma}
\begin{proposition} \label{lem:case2}
  For a subinterval $[t_{l}, t_{l+1})$, $l \in \{0,1,\ldots,\s-1\}$, where at least one rationalist and one imitative group have positive payoff values,
 the set of positive payoff values of the imitative groups $\mathcal{P}(t_l)$ is globally identifiable if  Assumptions \ref{ass:subintervals-vaccine-surplus-period} and \ref{ass:two-subintervals-shortage-surplus} are satisfied.
\end{proposition}
\begin{corollary} \label{cor:identifiability-of-the-cis-case2}
   Under the conditions of \Cref{lem:case2},
   the $(l+1)^{\text{th}}$ component of the vector $\bm \Delta \bm \pi_i$, for $i\in \mathcal{M}$, is
      identifiable if it is positive and the ordering of the positive payoff values of imitators is known.
\end{corollary}

Elementary symmetric polynomials in $n$ variables $c_1,c_2,\ldots,c_n$ are defined as $\sigma_j=0 $ for  $j<0$,  $\sigma_0=1$,
$\sigma_1 = \sum_{i=1}^n c_i$, $\sigma_2 = \sum_{1\leq j < k \leq n} c_jc_k$, and so forth, ending with $\sigma_n = c_1c_2\ldots c_n$. 
Define 
\begin{equation} \label{eq:denominator-betas}
   {B}(t,t_l) = \sum_{k=0}^{n_l} (-1)^k \beta_k(t) \sigma_{n_l-k}(t_l),
\end{equation}
where $\sigma_0(t_l), \ldots, \sigma_{n_l}(t_l)$  are elementary symmetric polynomials in members of the set $\mathcal{P}(t_l)$ defined in \eqref{eq:set-of-payoff}.
\begin{corollary} \label{cor:identifiability-xbar-IM-R}
Under the conditions of \Cref{lem:case2}, the proportion of unvaccinated individuals who, during the time interval $[t_l, t_{l+1})$, have non-positive payoff values, 
$\bar{x}(t_l)$,
is identifiable if there exists some subinterval $\mathcal{L}\in[t_l,t_{l+1})$ such  for all $t\in\mathcal{L}$ we have ${B}(t,t_l)\neq 0.$
\end{corollary}

Thus far, we have investigated the identifiability of the payoff values of the groups within a subinterval.
For the imitative groups it has been shown that the set of positive payoff values for each subinterval $[t_l,t_{l+1})$, $\mathcal{P}(t_l)$, is identifiable.
However, revealing the sets are not very informative unless we can associate their members over $[0,T)$.
One may assume \textit{a priori} known ordering of the payoff values of the imitative groups.
Otherwise, in the sequel we provide sufficient conditions for concatenation of the imitative payoff values throughout the entire $[t^*,T)$ and, in turn, $[0,T)$.
Denote by $\mathtt{r}(t_l)$ the number of rationalist groups with positive payoff values for the time interval $[t_l, t_{l+1})$.
 As it is shown in the proof of \Cref{lem:case2}, 
 for $\mathtt{r}(t_l) \neq 0$, the proportion of unvaccinated rationalists with positive payoff values for the time interval $[t_l,t_{l+1})$, $l \in \{0,1,\ldots,\s-1\}$, is equal to:
 \begin{equation} \label{eq:xprime-mainText}
    x'(t) =  \frac{\sum_{k=0}^{n_l} (-1)^k \alpha_k(t) \sigma_{n_l-k}(t_l) -\sigma_{n_l}(t_l) \bar{x}(t_l)}{{B}(t,t_l) },
 \end{equation}
 provided that    ${B}(t,t_l) \neq 0$.
The proportion of unvaccinated imitators with positive payoff value $\hat{ \pi}_i(t_l) \in \mathcal{P}(t_l)$ for $i \in \{1,\ldots,n_l\}$ and the time interval $[t_l, t_{l+1})$ also equals:
 \begin{equation} \label{eq:gamma-im-br-main-text}
 \scalebox{0.93}{$
    \begin{aligned}
        \hat{x}_i& (t)=   \frac{1}{{\!\displaystyle\prod_{j \neq i} }\big(\hat{\pi}_i(t_l) - \hat {\pi}_j(t_l)\big)}  \Bigg( \!\sum_{k=1}^{n_l+1} \hat{ \pi}_i^{n_l-k}(t_l) \bigg( \!\sum_{q=0}^{n_l} (-1)^{k+1+q}\alpha_q(t_l) \sigma_{k-1-q}(t_l) \bigg)\Bigg) \\
        &  - \1\big(\mathtt{r}(t_l)\big) x'(t_l) \frac{1}{{\displaystyle\prod_{j \neq i} }\big(\hat{\pi}_i(t_l) - \hat{\pi}_j(t_l)\big)}  
         \times \Bigg( \!\sum_{k=1}^{n_l+1} \hat {\pi}_i^{n_l-k}(t_l) \bigg( \!\sum_{q=0}^{n_l} (-1)^{k+1+q}\beta_q(t_l) \sigma_{k-1-q}(t_l) \bigg)\Bigg).
    \end{aligned}
    $}
\end{equation}
For two vectors of the same dimension, let
$\max (\bm v_1, \bm v_2)$ denote the vector obtained by applying the element-wise maximum between the corresponding components of $\bm v_1$ and $\bm v_2$.
Define 
\begin{equation}
    \begin{aligned}
        \bm {\mathcal{N}} = \big\{\big( x_i(t^*), \max (\bm \Delta \bm \pi_i, \bm 0)\big) \vert 
        \exists i\in \mathcal{M}, l \in \{0,\ldots,\s-1\}, \text{ s.t. } \Delta \pi_i(t_l) > 0  \big\}.
    \end{aligned}
\end{equation}

Each member of the set $ \bm {\mathcal{N}}$ is a tuple where the first element is the proportion of unvaccinated individuals at $t^*$ in an imitative group, and the second element is the vector of payoff values of that group, with  negative components set to zero.
If we can find the members of the set $\bm {\mathcal{N}}$, the trajectories of the  payoff values and the unvaccinated proportions of individuals in imitative groups will be revealed.
In the following, we provide sufficient conditions for this. 
\begin{proposition} \label{lem:gi-xstar}
    Assume that
     for each subinterval $[t_l, t_{l+1})$, $l \in \{0,1,\ldots,\s-1\}$,
    \emph{(i)} 
   Of the imitative groups, $n_l$
  have distinct positive payoffs, while the rest have non-positive payoffs,
\emph{(ii)}
if
at least one rationalist group has a positive payoff value and $n_l>0$, then
 ${B}(t_l,t_l) \neq 0,$
and
 \emph{(iii)} if $n_l >0$, for $i,j \in \{1,2,\ldots,n_l\}$, $j \neq i$, we have
$\hat{x}_i(t_l) \neq \hat{x}_j(t_l)$.
Then, under \Cref{ass:subintervals-vaccine-surplus-period} and \Cref{ass:two-subintervals-shortage-surplus}, the
set  $\bm {\mathcal{N}}$
is identifiable.
\end{proposition}


\subsection{Identifiability analysis of the  proportions of the groups} \label{subsec:identifiability-proportions}
The  proportion of unvaccinated individuals in group $i$, $i \in \{1,2,\ldots,\n\},$ at time $T_0 = 0$ equals the  proportion of this group, $\rho_i$.
In this section, we investigate the identifiability of the proportions $\rho_i$, $i \in \{1,2,\ldots, \n\},$ in two steps.
First, we provide sufficient conditions for the identifiability of the proportions of unvaccinated individuals in each group at vaccine abundance time $x_i(t^*)$ for $i \in \{1,\ldots,\n \}$.
Then, provided that the vector of these proportions $\x(t^*)$ is identifiable, the identifiability of its initial value $\x(0) = (\rho_1,\ldots, \rho_{\n})^\top$ is investigated.
 \subsubsection{Identifiability analysis of the proportions of unvaccinated individuals at $t^*$} \label{sub:Identifiability analysis of the proportions of unvaccinated individuals at}
  We would like to investigate the identifiability of the proportion of unvaccinated individuals in each group at  vaccine abundance time $x_i(t^*)$, for $i \in \{1,2,\ldots,\n\}$.
  To do so, we write the proportions of unvaccinated individuals at time $t$ in terms of their values at vaccine abundance time, the output and its higher-order derivatives.
Define matrix $\matr{A}$ as follows:
    \begin{equation} \label{eq:MatrixA-R-M}
\matr{A} =
    \begin{pmatrix}
       \eta_{1}(\Tilde{t}_1) & \eta_{2}(\Tilde{t}_1)  & \cdots & \eta_{\n}(\Tilde{t}_1)  \\
       \vdots & \vdots & \ddots & \vdots \\
        \eta_{1}(\Tilde{t}_m) & \eta_{2}(\Tilde{t}_m)  & \cdots & \eta_{\n}(\Tilde{t}_m)
    \end{pmatrix},
\end{equation}
 where $\Tilde{t}_j \in [t^*,T)$,  $j = 1,2,\ldots, m$, and 
\begin{equation} \label{eq:Definition-of-eta}
\eta_k(t) = 
    \begin{cases}
        \exp \bigg(-\kappa \displaystyle\sum_{j=0}^{q(t)} (t_j - t_{j-1}) \1 \big(\Delta \pi_k(t_{j-1})\big) \\
      \displaystyle   \qquad \quad -\kappa(t - t_{q(t)})\1\big(\Delta \pi_k(t_{q(t)})\big)  \bigg), & \text{ if } k \in \mathcal{R}, \\
        \exp \bigg(-\kappa \displaystyle\sum_{j=0}^{q(t)} \big(\Delta \pi_k(t_{j-1}) \1 \big(\Delta \pi_k(t_{j-1})\big)\int_{t_{j-1}}^{t_{j}} y(\tau)d\tau \big) \\ \qquad \quad  -\kappa
    \displaystyle \big(\Delta \pi_k(t_{q(t)}) \1\big(\Delta \pi_k(t_{q(t)})\big)\int_{t_{q(t)}}^{t} y(\tau)d\tau \big)
     \bigg),  & \text{ if } k \in \mathcal{M},
    \end{cases}
\end{equation}
where $t \in [t^*,T)$,       $ q(t) = \max \{ j \vert t \geq t_j \},$ and $t_j = t^*$ for $j <0$.

The $i^{\text{th}}$ row of matrix $\matr{A}$ corresponds to a time instant $\Tilde{t}_i$ in $[t^*,T)$, and the $k^{\text{th}}$ column of matrix $\matr{A}$ corresponds to group $k$, for $k \in \{1,2,\ldots, \n\}$.
\begin{lemma} \label{lem:ij-entry-matrix-A}
Under Assumptions \ref{ass:subintervals-vaccine-surplus-period} and \ref{ass:two-subintervals-shortage-surplus},
    the $(i,k)$ entry of matrix $\matr{A}$ equals the ratio of the proportion of unvaccinated individuals in group $k$ at time instant $\Tilde{t}_i$ to that at vaccine abundance time, i.e., $x_k(\Tilde{t}_i)/x_k(t^*)$.
\end{lemma}
 
In view of \Cref{lem:ij-entry-matrix-A}, the dot product of the $i^{\text{th}}$ row of matrix $\matr{A}$ and $\x(t^*)$ equals the proportion of unvaccinated individuals at time instant $\Tilde{t}_i$, i.e., $1-y(\Tilde{t}_i)$, which is known.
Thus, the identifiability of the proportions of unvaccinated individuals in each group at vaccine abundance time $\x(t^*)$ stems from matrix $\matr{A}$ having identifiable entries and being of full rank. 
The first condition requires information on the signs of all payoffs and the ordering of the positive payoff values of imitators.
 \begin{proposition} \label{prop:identifiability-of-x-tstar}
    Under Assumptions \ref{ass:subintervals-vaccine-surplus-period} and \ref{ass:two-subintervals-shortage-surplus},
     assume that there exist
     $m \geq \n$ time instants $\Tilde{t}_j \in [t^*, T)$ 
     such that the signs of the payoff values of all $\n$ groups and the ordering of the positive payoff values of imitative groups are known for all $t \leq \Tilde{t}_j$.
     Then,
    the proportions of unvaccinated individuals in each group at vaccine abundance time $t^*$ are globally identifiable if  matrix $\matr{A}$, defined in \eqref{eq:MatrixA-R-M},  is of full rank.
 \end{proposition} 
 \begin{remark} \label{rem:multiple-refusers-rationalists}
  In view of \Cref{lem:ij-entry-matrix-A}, all elements in the columns associated with vaccine refusing groups are one.
Hence, in the presence of multiple vaccine-refusing groups, matrix 
$\matr{A}$ is not full rank.
In this case, we can aggregate these groups and replace them with a single vaccine-refusing group whose proportion equals the sum of the proportions of all the vaccine-refusing groups.
A similar reasoning holds  for rationalist groups where the corresponding components of their payoff vectors have the same signs.
 \end{remark}


What if we do not know the ordering of the positive payoff values of the imitators? In the following, we provide an alternative. 
Recall that the cardinality of the set $\bm{\mathcal{N}}$, denoted by $\vert \bm{\mathcal{N}} \vert$, equals  the number of non-vaccine  refusing imitative groups.
Now, construct a matrix  $\matr{A}^*$ similar to matrix $\matr{A}$, but with $\n-\vert \bm{\mathcal{N}} \vert $ columns, where each column
is
associated with either a rationalist group or  the vaccine-refusing  group.
\begin{proposition} \label{prop:identifiability-x-star-2}
    Assume  there
    are $\vert \bm{\mathcal{N}} \vert$ non-vaccine refusing imitative groups and
     $m \geq \n - \vert \bm{\mathcal{N}} \vert$ time instants $\Tilde{t}_j \in [t^*, T)$, $j = 1,2,\ldots,m$, 
     such that the signs of the payoff values of all rationalist groups are known for all $t \leq \Tilde{t}_j$.
     Then, under the conditions of \Cref{lem:gi-xstar},
    the proportions of unvaccinated individuals in $\n$ groups at vaccine abundance time $t^*$ are globally identifiable if  matrix $\matr{A}^*$   is of full rank,
    the mapping from the elements of the set $\bm{\mathcal{N}}$ to the elements of the set $\mathcal{M}$ is a known function,
    and Assumptions \ref{ass:subintervals-vaccine-surplus-period} and \ref{ass:two-subintervals-shortage-surplus} are satisfied.
\end{proposition}
\subsubsection{
 Identifiability of $\rho_i$}

Define $\mathcal{X}$ as
$
\prod_{j=1}^{{\n}}(0,1).
$
During the vaccine shortage period $t \in [0,t^*)$, the evolution of $\x(t) = \big(x_1(t), \ldots, x_\n(t)\big)$ reads as follows
\begin{equation}\label{eq:gammaVector}
\begin{aligned} 
 &\dot{\bm x}(t) = \bm h(t, \bm x),\\
 \bm h:& [0,t^*) \times \mathcal{X} \to \mathcal{X}, 
\end{aligned}
\end{equation}
where the $i^{\text{th}}$ component of $\bm h$, $\dot{x}_i$, equals
\begin{gather} \label{eq:dotGammaBRandIM}
 \dot{x}_i(t) =
    \begin{cases}
   \frac{-\kappa x_i(t) \1\big(\Delta \pi_i(t)\big)u(t)}{\sum_{j \in \mathcal{M}} x_j(t)\big(1 - \sum_{q=1}^\n x_q(t) \big)\Delta \pi_j(t)\1\big(\Delta \pi_j(t)\big) + \sum_{j \in \mathcal{R} }x_j(t) \1\big(\Delta \pi_j(t)\big)}, & \text{if } i \in \mathcal{R},\\
   \frac{-\kappa x_i(t) \big(1 - \sum_{q=1}^\n x_q(t) \big)\Delta \pi_i(t)\1\big(\Delta \pi_i(t)\big)u(t)}{\sum_{j \in \mathcal{M}} x_j(t)\big(1 - \sum_{q=1}^\n x_q(t) \big)\Delta \pi_j(t)\1\big(\Delta \pi_j(t)\big) + \sum_{j \in \mathcal{R}}x_j(t) \1\big(\Delta \pi_j(t)\big)}, & \text{if } i \in \mathcal{M}.
    \end{cases}
\end{gather}
The value of $x_i(t)$ at $t=0$ equals $\rho_i$.
If the value of $x_i(t^*)$  is identifiable, 
then, in view of \eqref{eq:dotGammaBRandIM}, unidentifiability of the proportion $\rho_i$ implies that
different initial values of $x_i(t)$ at $t=0$ result in the same value for $x_i(t)$ at $t=t^*$. 

\begin{lemma} \label{prop:identifiability-of-x0}
     Assume that for the time interval $[0,t^*)$
     at least one rationalist group has positive payoff value
     and  input $u(t)$ is smooth.
     Then, the initial condition of the  dynamical system     \eqref{eq:gammaVector} is globally identifiable
     if either 
       the conditions in \Cref{prop:identifiability-of-x-tstar} or
       \Cref{prop:identifiability-x-star-2} are satisfied.
\end{lemma}
The main result of this paper is summarized in the following theorem.
\begin{theorem} \label{thm:the-final-theorem}
    Consider a population consisting of $\n$ groups of decision-making individuals with vaccination dynamics model \eqref{eqn:vax-seeker} and \eqref{eqn:mean-dynamics}, and the  cumulative proportion of vaccinated individuals as the output.
    Under the conditions of \Cref{prop:identifiability-of-x0},
    the  proportions of the $\n$ groups, that is $\rho_1, \ldots, \rho_\n$, are globally identifiable. 
\end{theorem}
\begin{remark}
If the mapping from the elements of the set $\bm{\mathcal{N}}$ to the elements of the set $\mathcal{M}$ and the ordering of the positive payoff values of imitators are not known,
then, at best,  the set of proportions of the imitative groups is identifiable, without being able to associate them to each imitative group.
  \end{remark}
\section{On Identifiability of the set $\mathcal{T}$}
The obtained results in the previous sections were based on the assumption that the endpoint time instants $T_1, \ldots, T_{\s-1}$ and, in turn, the set $\mathcal{T}$ are known.
In this section
we provide a sufficient condition for the identifiability of some endpoint time instant $\hat{t} \in \{T_1, \ldots, T_{\s-1}\}$ at which the payoff values of some imitative groups change.

The idea is that during each subinterval the payoff values for the groups remain constant and, consequently the output is arbitrarily differentiable.
Therefore, if there are time instants where the output is not arbitrarily differentiable, they correspond to the endpoints of some subintervals.
By $\Delta \pi_i(\hat{t}^-)$ (resp. $\Delta \pi_i(\hat{t}^+)$), we mean the left-hand limit $\lim_{t \to \hat{t}^-}\Delta \pi_i(t)$ (resp. right-hand limit $\lim_{t \to \hat{t}^+}\Delta \pi_i(t)$).
Denote the homogeneous polynomial of degree $l$ in two variables $\Delta \pi_i(\hat{t}^-) \1\big(\Delta \pi_i(\hat{t}^-) \big)$ and 
     $\Delta \pi_i(\hat{t}^+) \1\big(\Delta \pi_i(\hat{t}^+) \big)$ with all coefficients equal to one by $p_l^{[i]}$, i.e., 
     $$p_l^{[i]} = \sum_{j=0}^{l} \Delta \pi^j_i(\hat{t}^-) \1\big( \Delta \pi^j_i(\hat{t}^-)\big)\Delta \pi^{(l-j)}_i(\hat{t}^+) \1\big(\Delta \pi^{(l-j)}_i(\hat{t}^+) \big).$$
\begin{lemma} \label{lem:t-existence-imitation} 
    Let the payoff values of $\m$ imitative groups, $\{m_1,m_2,\ldots, m_\m\} \in \mathcal{M}$, change at
   time instant $\hat{t}$ such that $\1\big(\Delta \pi_i(\hat{t}^-) \big)\1 \big(\Delta \pi_i(\hat{t}^+) \big) >0$ for some $i \in \{m_1,m_2,\ldots, m_\m\}$.
     The output trajectory $y(t)$ is then at most $\m$-times differentiable at time instant $\hat{t}$
    for all payoff values $\Delta \pi_i(\hat{t}^-)$ and 
     $\Delta \pi_i(\hat{t}^+)$, $i = m_1,\ldots,m_\m$,
     except possibly for the sets of payoff values characterized by  $\text{det }(\matr{P})= 0$ where
    \begin{equation}\label{eq:matrix-xi}
    \matr{P} = 
        \begin{pmatrix}
       1 & 1  & \cdots & 1   \\
       p^{[m_1]}_1 & p^{[m_2]}_1 & \cdots & p^{[m_\m]}_1\\
       \vdots & \vdots & \ddots & \vdots   \\
        p^{[m_1]}_{\m-1} & p^{[m_2]}_{\m-1} & \cdots & p^{[m_\m]}_{\m-1}
    \end{pmatrix}.
    \end{equation}
\end{lemma}
The condition $\1\big(\Delta \pi_i(\hat{t}^-) \big)\1 +\big(\Delta \pi_i(\hat{t}^+) \big) >0$
excludes the case where the payoff values of all imitative groups $i$ change between two non-positive values.
\begin{proposition} \label{prop:t-discontinuous-imitation} 
    The time instant $\hat{t}$ at which the payoff values of $\m$ imitative groups, $\{m_1,m_2,\ldots, m_\m\} \in \mathcal{M}$, change is identifiable if $\text{det }(\matr{P}) \neq 0$ and $\1\big(\Delta \pi_i(\hat{t}^-) \big)$ $\times\1 \big(\Delta \pi_i(\hat{t}^+) \big) >0$ for some $i \in \{m_1,m_2,\ldots, m_\m\}$.
\end{proposition}

The result relies on the assumption of having access to  higher-order derivatives of the output at each time instant which limits its practical applicability.
\section{A note on practical identifiability}
In this paper, we were interested in
inherently identifiable parameters  in the developed  vaccination behavior model  \eqref{eqn:vax-seeker} and \eqref{eqn:mean-dynamics}, i.e.,
\emph{a priori} identifiability.
Accordingly, the provided results rely on the key assumptions of an error-free model with perfect and noise-free data. 
These assumptions highlight the distinction between practical and theoretical identifiability, meaning that theoretical identifiability does not guarantee ``well-determined" estimation results.
The analysis of practical identifiability for the system under study is beyond the scope of this paper. 
For a  discussion on this topic, please refer to \cite{wieland2021structural,lam2022practical}.

As for the number of required data points, in our system, if there are $\vert \mathcal{M} \vert$ imitative groups, the highest order derivative of the output that appears in the input-output relation is $\vert \mathcal{M} \vert + 1$.
To calculate this order using finite differences, a minimum of $\vert \mathcal{M} \vert + 2$ data samples is required. 
The number of parameters to be identified from this input-output relation is $\vert \mathcal{M} \vert + 1$, implying that at least $2\vert \mathcal{M} \vert + 2$ data points are necessary for each interval \cite{miao2011identifiability}. 
However, as vaccination progress data is usually daily, accurately approximating higher-order derivatives of the output becomes challenging.
\section{Revisiting \Cref{example:first}}

The plot in 
\Cref{fig:example-introduction} has been depicted based on the following parameter values: the proportions of seniors who are vaccine refusers, imitative, and rationalist  are respectively equal to $0.009$, $0.15$, and $0.01$.
As for non-seniors, the proportions of vaccine refusers, imitators, and rationalists are respectively equal $0.2$, $0.25$, and $0.381$.
The payoff values for imitative seniors are assumed to be $0.67$ and $0.9$ for the time intervals  $t \in [0, 20)$ and $t \in [20,40)$, respectively.
For imitative non-seniors, the corresponding payoff values are assumed to be $0.23$ and $0.5$, respectively.
The payoff values for rationalist (resp. vaccine refusing) seniors are assumed to be $0.15$ and $0.2$ (resp.  $-0.2$ and $-0.25$) for the time intervals  $t \in [0, 20)$ and $t \in [20,40)$, respectively.
For rationalist  (resp. vaccine refusing) non-seniors, the corresponding payoff values are assumed to be $0.05$ and $0.25$ (resp.  $-0.15$ and $-0.15$), respectively.
The maximum rate of vaccination $\kappa$ is set to $0.1.
$

Identifiability analysis led to the identifiability of the maximum rate of vaccination, $\kappa$, the sets of imitators' payoff values, $\mathcal{P}(0)$ and $\mathcal{P}(20)$,  the total proportion of rationalist, $\rho_1 + \rho_2$, and the total proportion of vaccine-refusers $\rho_5 + \rho_6$.
As detailed in \Cref{sec:further}, by assuming a known ordering among the payoff values of the imitative seniors and non-seniors, their payoff values and, in turn, their proportions are identifiable as well.
Otherwise, since the conditions in \Cref{lem:gi-xstar} are met, the set $\bm{\mathcal{N}}$ is identifiable (\Cref{fig:enter-label}).

We simulated the model with true parameter values and recorded  the cumulative proportion of vaccinated individuals $y(t)$, and its derivatives up to the third order at six time instants.
As detailed in the \Cref{sec:further}, we then estimated the parameters $\kappa$, $\rho_1 + \rho_2$, $\rho_4 + \rho_5$, $\rho_3$, $\rho_4$, $\pi_n(0)$, $\pi_n(20)$, $\pi_s(0)$, and $\pi_s(20)$ and the estimation results are reported in the second row of \Cref{tab:example2}.

 We additionally simulated the model with true parameters and recorded the cumulative proportion of vaccinated individuals,  now at $12-$hour intervals.
Following \cite{tuncer2018structural, eisenberg2013identifiability}, we synthesized noisy data by adding  Gaussian error with  a standard deviation equal to $0.1$ of the mean to the recorded data from the model with true parameters' values.
Gaussian error can be interpreted as a measurement error.
Following \cite{chowell2017fitting, eisenberg2013identifiability}, we also considered Poisson error structure where we assumed  a population of $10^6$ individuals and modeled the vaccine uptake per each time unit as a Poisson process.
Poisson noise structure can indeed account for the process noise in our set-up.
The vaccine uptake by individuals is a random process--under the assumption of independency and identically distributed across each time interval. 
Therefore, the total number of vaccinated individuals within a fixed time interval can be modeled by a Poisson distribution.
We then ran the simulation $1000$ times with different realizations of the noisy data and initial parameter guesses.
We estimated the parameters by
using $\mathtt{curve\_fit}$ function in Python $\mathtt{scipy}$ package.
The valid intervals for all parameters were set to $[0,1].$
\Cref{tab:example2} reports the mean and the coefficient of variation ($100\%\times$ standard deviation divided by the mean).
The violin plots are reported in \Cref{fig:violin}. 
We additionally ran the parameter estimation described earlier, but now without using information on the ordering of the payoff values.
The result of the parameter estimation and 
the  violin plots are given in  \Cref{tab:example2-noordering} and
\Cref{fig:violin-no-ordering}.
\begin{table}  
\caption{The true and estimated parameter values in \Cref{example:first}.
The proportions of seniors who are rationalists, imitators, and vaccine-refusers are denoted by $\rho_1$, $\rho_3$, and $\rho_5$, respectively.
The corresponding proportions for the non-seniors are denoted by $\rho_2$, $\rho_4$, and $\rho_6$, respectively.
   The payoff value of imitative seniors (resp. non-seniors) is denoted by $\Delta \pi_s$ (resp. $\Delta \pi_n$).
 The fourth (resp. fifth) row reports the means and the coefficients of variation, denoted by CV, of the parameters obtained from fitting the model to  $1000$ simulated  data, which were made noisy with Gaussian noise having a mean of zero and a standard deviation equal to ten percent of the mean (resp. Poisson noise structure). 
The initial parameter guesses were drawn from a uniform distribution within $50\%$ of the true values.
}\label{tab:example2}
\centering 
\resizebox{\textwidth}{!}{
\begin{tabular}{|l|r|r|r|r|r|r|r|r|r|r|} 
\hline
  &  $\kappa$ & $\rho_1+ \rho_2$ & $\rho_3$ & $\rho_4$ & $\rho_5 + \rho_6$ & $\Delta \pi_s(0)$ & $\Delta \pi_n(0)$  & $\Delta \pi_s(t_1)$ & $\Delta \pi_n(t_1)$  \\ 
  \hline
  \footnotesize \makecell[l]{True value}  &  $0.1$& $0.391$ & $0.15$ & $0.25$ & $0.209$ & $0.67$ & $0.23$ & $0.9$ & $0.5$  \\
  \hline 
 \footnotesize \makecell[l]{Estimated value \\ (no noise)}  
 & $0.1$& $0.393$ & $0.144$ & $0.254$ & $0.209$& $0.67$ &$0.23$  &$0.9$  & $0.5$ 
   \\
   \hline
         \footnotesize\makecell[l]{Mean (Gaussian) \\ CV }   & \makecell{$0.105$\\$19\%$}& \makecell[r]{$0.386$\\$16\%$} & \makecell[r]{$0.148$\\$33\%$} & \makecell[r]{$0.271$\\$22\%$} & \makecell[r]{$0.195$\\$31\%$} & \makecell[r]{$0.69$\\$39\%$} & \makecell[r]{$0.22$\\$22\%$} & \makecell[r]{$0.91$\\$22\%$} & \makecell[r]{$0.476$\\$24\%$} 
         \\
    \hline
       \footnotesize   \makecell[l]{Mean (Poisson) \\ CV}   & \makecell{$0.105$\\$19\%$}& \makecell[r]{$0.384$\\$16\%$} & \makecell[r]{$0.149$\\$39\%$} & \makecell[r]{$0.265$\\$23\%$} & \makecell[r]{$0.201$\\$28\%$} & \makecell[r]{$0.683$\\$40\%$} & \makecell[r]{$0.237$\\$38\%$} & \makecell[r]{$0.91$\\$20\%$} & \makecell[r]{$0.483$\\$24\%$} \\
         \hline
\end{tabular}
}
\end{table}
\begin{figure}
    \centering
    \includegraphics[width=0.9\linewidth]{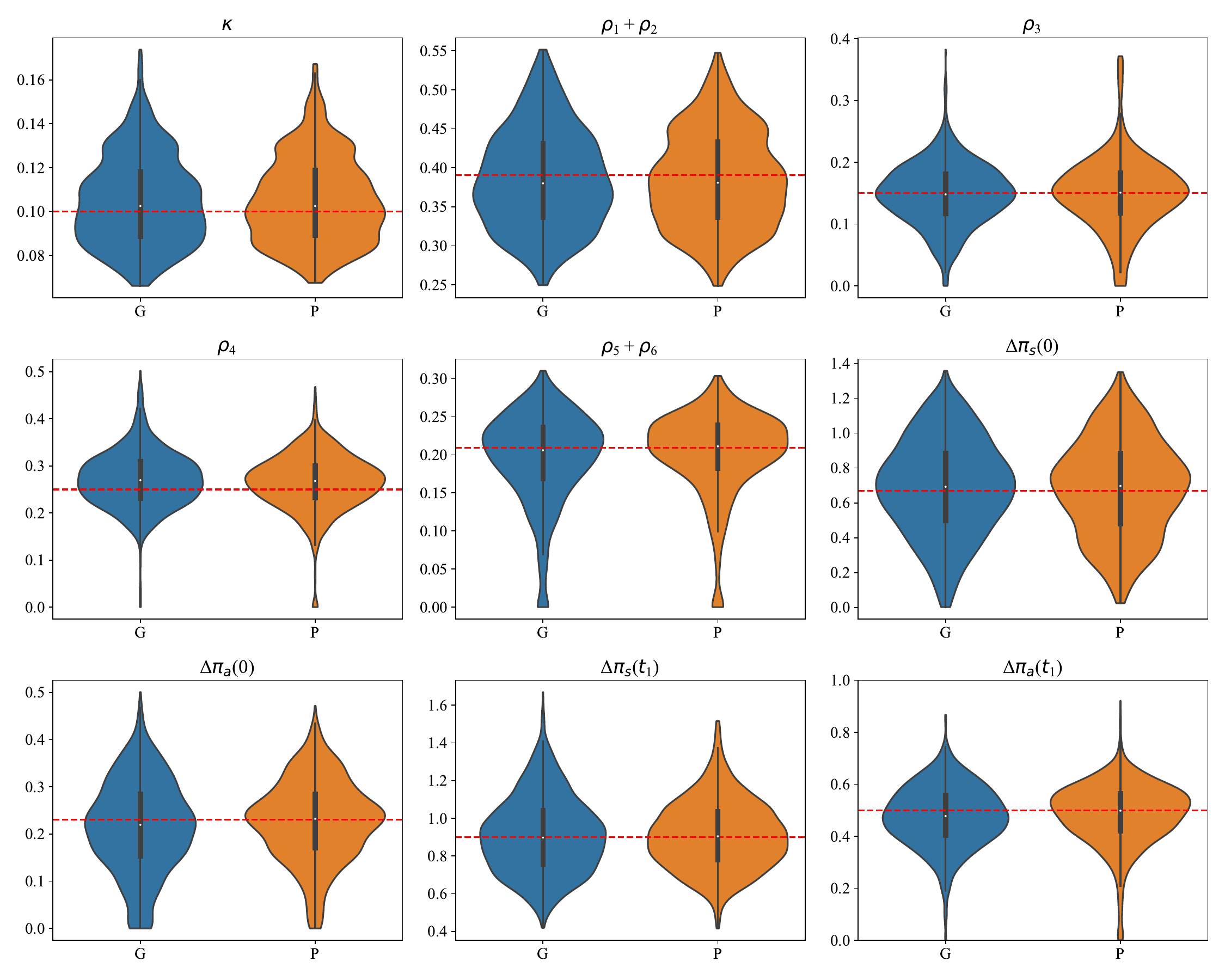}
    \caption{The violin plots of the parameter values estimated  from
    fitting the model to  $1000$ simulated noisy data.
   The proportions of rationalists and vaccine-refusers are denoted by $\rho_1 + \rho_2$ and $\rho_5 + \rho_6$, respectively.
   Those of imitative seniors and non-seniors, respectively, are denoted by $\rho_3$ and $\rho_4$.
   The payoff value of  imitative seniors (resp. non-seniors) is denoted by $\Delta \pi_s$ (resp. $\Delta \pi_n$).
    The letter ``G'' (resp. ``P'') stands for Gaussian (resp. Poisson) noise structure.
    The dashed lines show the true parameter values.
    The ordering of the payoff values was assumed to be known.
     }
    \label{fig:violin}
\end{figure}

The result of the simulation study suggests that the identifiable parameters can be estimated reliably in practice.

\section{Concluding Remarks}
We investigated the identifiability of the  proportions of individuals deciding to vaccinate in a large, well-mixed, and heterogeneous population. 
We started off by developing an ODE-based mechanistic model capturing the heterogeneity of individuals in terms of  the perceived payoff gains for vaccination and the decision-making strategies.
We modeled the perceived payoff gains for vaccination by  piecewise constant functions over time.
As for decision-making strategies, we assumed that individuals decide to vaccinate  either by evaluating their own perceived payoff gain for vaccination or based on their interaction with others and comparing their own payoff to those of others. 
In the literature, the former one is usually referred to as an innovator \cite{van2007new}, an asocial learner \cite{mesoudi2015higher}, a best-responder \cite{le2020heterogeneous}, a myopic rationalist, or an evidence-based learner.
The latter is referred to as an imitator \cite{van2007new}, a payoff-based learner \cite{molleman2014consistent}, a success-based learner \cite{van2015focus}, or a social learner \cite{mesoudi2015higher}.

We then investigated the identifiability of each parameter appearing in the developed model.
Our analysis provides conditions for the identifiability of the maximum rate of vaccination and implies  the unidentifiability of the rationalists' payoff gains.
As for payoff gains of imitators, we proved their identifiability under some conditions.
We then proved the identifiability of the  proportions of groups, based on the
 identifiability of the payoff gains of imitators and  the main theorem of uniqueness.


In our model, the aggregation of certain groups may represent a socioeconomic or demographic group, where the proportion of the aggregated group is known \emph{a priori} and can be used as an additional equality constraint on the proportions.

Several notions of identifiability  have been considered in the literature \cite{chis2011structural, anstett2020priori}.
These notions mainly differ in terms of initial conditions and input values.
Here, the knowledge of initial conditions is required for the identifiability of the proportions of decision-makers.
As for the input, the provided results on identifiability hold true for any input function as long as the input is smooth and splits the time span into two disjoint subintervals: one of supply limitation and one of abundance.

We provide sufficient conditions for identifiability of the trajectories of the imitators' payoffs over the course of vaccination program.
This finding could have the following practical implications: \emph{(i)} understanding the impact of each event on the payoff values of the imitative groups and \emph{(ii)} predicting the  changes in the payoffs  in response to similar events in future vaccine campaigns.



The theoretical identifiability is a perquisite for \emph{practical identifiability},
yet often it is neglected.
Using numerical optimization algorithms, we can estimate  values for each parameter in every mechanistic model including  \eqref{eq:dotGamma}. 
However, the estimated values could be misleading in the absence of theoretical identifiability analysis.

The results of this paper open the door to  reliable estimation of the proportions of decision-makers and their payoffs using longitudinal data on vaccine uptake, which remains a topic for future investigation.

\appendix
\section{Proofs} \label{sec:proof}
 \subsection{Some useful characteristics of the function $f_{(j,k)}$}
 \begin{lemma} \label{cor:f-k-k}
    Consider the function $f_{(j,k)}$ which is defined \eqref{eq:f-j-k}.
    The value of $f_{(k,k)}$ is  equal to
   $
   f_{(k,k)} = (-1)^{k+1} \kappa^k y^k.
$
\end{lemma}
\begin{lemma}\label{lem:f-j-k-derivative-output}
    $f_{(j,l)}$ for $j < l$ does not include any monomials without the  derivatives of the output $y(t)$.
\end{lemma}
\begin{corollary} \label{cor:alphas-M-R-have-derivatives}
    For $k > 0$, the  term $y^n\alpha_k$, where $\alpha_k$ is defined in \eqref{eq:alpha-series}, does not contain any monomials without the  derivatives of the output.
\end{corollary}
\begin{lemma} \label{lem:f-j-j+k-derivative-output}
    The term $f_{(j,j+l)}$ for $l>0$ does not include any monomials with derivatives of the output of the order ${(l+1)}$ or higher.
\end{lemma}
\subsection*{Proofs of \Cref{lem:case1} and \Cref{cor:identifiability-xbar-imitators}}
The proof of the following lemma, which is used in the proof of \Cref{lem:case1}, is immediate and hence omitted.
The term  $(-1)^{k+1}\kappa^k y^{k+s}\alpha_k$ is a summation of rational functions of the variable $y$ and its derivatives, and parameter $\kappa$.
The denominators in this term contain only the variable $y$, possibly raised to a non-negative integer.
Let $s$ be the highest degree of the variable $y$ which appears in the denominators of the term  $(-1)^{k+1}\kappa^k y^{k+s}\alpha_k$.
\begin{lemma} \label{lem:induction-alpha-contains-ydot}
The differential polynomial $(-1)^{k+1}\kappa^k y^{k+s}\alpha_k$ for
$k \in \Z_{>0}$  consists of monomials with the (higher-order) derivatives of the output. 
\end{lemma}
Now, we move on to the proof of 
Proofs of \Cref{lem:case1} and \Cref{cor:identifiability-xbar-imitators}.
For  readability, 
the dependency on $t_l$ in notations $\bar{x}(t_l)$ and $n_l$ is omitted.
Additionally, we use the notation $c_i$ instead of $\Delta \pi_i(t_l)$ for $i \in \{1,\ldots,n\}$.
As it is mentioned in \cite{man2017computing}, the inverse of matrix $\matr{V}$ can be written as 
$
    \matr{V}^{-1} = \matr{W}\matr{Z},
$
where
\begin{equation*}
\scalebox{0.9}{$
    \matr{W} = 
     \begin{pmatrix}
       \frac{c_1^{n-1}}{\Pi_{j\neq 1,n+1}(c_1-c_j)} & \frac{c_1^{n-2}}{\Pi_{j\neq 1,n+1}(c_1-c_j)}  & \cdots & \frac{1}{\Pi_{j\neq 1,n+1}(c_1-c_j)} & \frac{1}{c_1\Pi_{j\neq 1,n+1}(c_1-c_j)}  \\
       \frac{c_2^{n-1}}{\Pi_{j\neq 2,n+1}(c_2-c_j)} & \frac{c_2^{n-2}}{\Pi_{j\neq 2,n+1}(c_2-c_j)}  & \cdots & \frac{1}{\Pi_{j\neq 2,n+1}(c_2-c_j)} & \frac{1}{c_2\Pi_{j\neq 2,n+1}(c_2-c_j)} \\
       \vdots & \vdots & \vdots & \ddots & \vdots \\
       \frac{c_n^{n-1}}{\Pi_{j\neq n,n+1}(c_n-c_j)} & \frac{c_n^{n-2}}{\Pi_{j\neq n,n+1}(c_n-c_j)}  & \cdots & \frac{1}{\Pi_{j\neq n,n+1}(c_n-c_j)} & \frac{1}{c_n\Pi_{j\neq 1,n+1}(c_1-c_j)}
       \\
        0 & 0  &  0 & \cdots & \frac{(-1)^n}{\sigma_n} 
   \end{pmatrix},
   $}
\end{equation*}
and matrix $\matr{Z}$ is a lower triangular matrix defined by
\begin{equation*}
\scalebox{0.8}{$
    \matr{Z} = 
      \begin{pmatrix}
       1 & 0  & 0 & \cdots   & 0  \\
       -\sigma_1 & 1 & 0 & \cdots & 0\\
       \sigma_2 & -\sigma_1 & 1 & \cdots & 0\\
       \vdots & \vdots & \vdots & \ddots & \vdots \\
        (-1)^n\sigma_{n} & (-1)^{n-1}\sigma_{n-1}  &  \cdots & \cdots & 1  
    \end{pmatrix},
    $}
\end{equation*}
where
 $\sigma_k$ is defined by 
 $\sigma_k=0$ for a negative $k$,
 $\sigma_0=1$, $\sigma_1 = \sum_{i=1}^n c_i$, $\sigma_2 = \sum_{1\leq j < k \leq n} c_jc_k$, and so forth, ending with $\sigma_n = c_1c_2\ldots c_n$.
Vector $\hat{\bm x}$ will be
\begin{equation}
\scalebox{0.9}{$
\begin{pmatrix}
        \hat{x}_1(t)\\
        \hat{x}_2(t)\\
        \vdots \\
        \hat{x}_n(t)\\
        \Bar{x}
    \end{pmatrix} =
 \matr{W}
 \times 
    \begin{pmatrix}
       \alpha_0(t) \\
       \alpha_1(t)-\alpha_0(t)\sigma_1 \\
       \alpha_2(t) -\alpha_1(t)\sigma_1+\alpha_0(t)\sigma_2 \\
       \vdots  \\
      \alpha_{n}(t)  + \ldots +  \alpha_0(t)(-1)^n\sigma_n 
   \end{pmatrix}.
   $}
\end{equation}
In particular, for $i \neq n+1$,
\begin{equation}
        \hat{x}_i(t)= \frac{1}{\prod_{j \neq i,n+1} (c_i - c_j)} \Bigg( \sum_{k=1}^{n+1} c_i^{n-k} \bigg( \sum_{l=0}^{n} (-1)^{k+1+l}\alpha_l(t) \sigma_{k-1-l} \bigg)\Bigg),
\end{equation}
\begin{equation}\label{eq:IO-sigmaBar}
    \Bar{x} = \frac{1}{\sigma_n} 
    \sum_{l=0}^n (-1)^l \alpha_l(t) \sigma_{n-l}.
\end{equation}
The proportion of unvaccinated individuals with non-positive payoff values $\Bar{x}$ for $t\in (t_l,t_{l+1})$ is constant.
Hence, it can be considered as an additional unknown parameter.
With that said, \eqref{eq:IO-sigmaBar} is the input-output relation.
It is straightforward to show that  the highest degree of $y$ appearing in the denominator of the term $\alpha_n$, denoted by $D(\alpha_n)$ does not exceed $n + D(\alpha_{n-1})$ where $D(\alpha_1) = 1.$ 
For the ease of notation we refer to $D(\alpha_n)$ by $m$.
By multiplying both sides of \eqref{eq:IO-sigmaBar} by $\sigma_n y^m$, 
\eqref{eq:IO-sigmaBar} turns into a polynomial differential equation in the following form:
\begin{equation} \label{eq:IOpolynomial}
   - \sigma_n \Bar{x} y^m + y^m\sigma_n {(1-y)} +
    \sum_{l=1}^n (-1)^l \alpha_l y^m\sigma_{n-l} = 0.
\end{equation}
In view of  \eqref{eq:alpha-series}, it can be shown that the highest order derivative of $y$, i.e., $y^{(n)}$ in \eqref{eq:IOpolynomial} appears only in $\alpha_n$ in the form of $y^{(n)}/\big((-1)^{n+1}y^n\kappa^n\big).$
To make \eqref{eq:IOpolynomial} normalized, we hence multiply both sides by $\kappa^n$, resulting in 
$   - \kappa^n\sigma_n \Bar{x} y^m + y^m\sigma_n\kappa^n (1-y)  -y^{m-n}y^{(n)}  +
    \sum_{l=1}^{n-1} \kappa^n(-1)^l \alpha_l y^m\sigma_{n-l}= 0.$
According to \Cref{thm:Eisenberg}, the coefficients of an input-output equation for a single-output system are identifiable.
By substituting \eqref{eq:alpha-series} in
$   - \kappa^n\sigma_n \Bar{x} y^m + y^m\sigma_n\kappa^n (1-y)  -y^{m-n}y^{(n)}  +
    \sum_{l=1}^{n-1} \kappa^n(-1)^l \alpha_l y^m\sigma_{n-l}= 0$
, we have 
\begin{equation} \label{eq:inmiddle}
\begin{aligned}
   \kappa^n\sigma_n (1-\Bar{x}) y^m -\sigma_n\kappa^n  y^{m+1}  -y^{m-n}y^{(n)} 
  - \sum_{l=1}^{n-1} \kappa^{n-l}  
   y^{(l)}    
   y^{m-l}\sigma_{n-l} 
   &\\ +
   \sum_{l=1}^{n-1} \kappa^{n-l} y^{m-l}\sigma_{n-l}
   (\sum_{j=1}^{l-1} f_{(j,l)} \alpha_j)   
   =& 0.
\end{aligned}
\end{equation}
In view of \Cref{lem:f-j-k-derivative-output}, $f_{(j,l)}$ for $j<l$ consists of monomials with  derivatives of the output.
This also holds for $\alpha_j$ for $j>1$( \Cref{lem:induction-alpha-contains-ydot}).
Therefore, each monomial of the term $f_{(j,l)}\alpha_j$ contains at least two components of (higher-order) derivatives of the output multiplied by each other, where one component is due to $f_{(j,l)}$ and the other is due to $\alpha_j$.
As a result, 
in the polynomial differential equation \eqref{eq:inmiddle}, the monomial $y^{m-l}y^{(l)}$, for $l \in \{1,2,\ldots,n-1\}$,  appears  with coefficient $\kappa^{n-l}  \sigma_{n-l}$.
With this and assuming the identifiability of $\kappa$, the terms $\sigma_i$ for $i = 1,\ldots,n-1$ are  identifiable.
In view of \Cref{lem:induction-alpha-contains-ydot}, other than the first two terms of the equation \eqref{eq:inmiddle}, the remaining terms contain monomials with the (higher-order) derivatives of the output.
This results in the identifiability of the coefficient of 
$ y^m$, i.e., $\kappa^n\sigma_n (1-\Bar{x})$ and that of
$y^{m+1}$, i.e., $\kappa^n \sigma_n.$
As a result besides $\sigma_n$, the population proportion of individuals with non-positive payoff value over this time interval, i.e., $\bar{x}$ is identifiable.
The terms $\sigma_i$ for $i = 1,\ldots,n-1$ are indeed the coefficients of a polynomial equation of order $n$ in the form $\sum_{i=0}^n (-1)^i\sigma_i x^{n-i}=0$ whose set of solutions is $\mathcal{P}_l$.
\subsection*{Proof of \Cref{cor:identifiability-of-the-cis-imitators}}
The proof is immediate.
 \subsection*{Proof of \Cref{lem:ci-gammai-based-on-alphas-betas}}
 If $t^* = 0$ or the vaccine supply is identically zero over $[0,t^*)$, we consider $[\epsilon, t_1)$ for $l=0$ and an arbitrarily small positive $\epsilon$.
    We prove by strong induction.
    For the sake of readability we drop the dependency on $t_l$.
    Base case: for $k=2$,
    in view of \eqref{eq:y-and-its-derivatives-R-and-M}, we have 
    \begin{equation} \label{eq:for-proof-alpha-beta}
    \scalebox{0.9}{$
     \sum_{i\in \mathcal{M}} \Delta \pi_i^2 \1 \big(\Delta \pi_i\big) x_i = \frac{\ddot{y} - f(1,2)\sum_{i \in \mathcal{M}} \
\Delta \pi_i \1\big(\Delta \pi_i\big) x_i +  \kappa^2 x'}{f_{(2,2)}},
$}
     \end{equation}
     In view of \eqref{eq:y-and-its-derivatives-R-and-M}, 
     the term $ 
     \sum_{i\in \mathcal{M}} \Delta \pi_i \1 \big(\Delta \pi_i\big) x_i$ is equal to $ \frac{\dot{y}}{\kappa y} - x'\frac{1}{y},$ which can be rewritten as $\alpha_1 - x'\beta_1.
   $
     By plugging the term 
     $\alpha_1 - x'\beta_1
   $
     into \eqref{eq:for-proof-alpha-beta}, we have
$\sum_{i\in \mathcal{M}} \Delta \pi_i^2 \1 \big(\Delta \pi_i\big) x_i = \big( \ddot{y} - f_{(1,2)}(\alpha_1 - x' \beta_1) +  \kappa^2 x' \big) /f_{(2,2)}.$
Indeed, the term $\big(\ddot{y} -  f_{(1,2)}\alpha_1\big)/f_{(2,2)}$ is equal to $\alpha_2$, and the term
     $-\big(   f_{(1,2)}\beta_1 + \kappa^2\big)/f_{(2,2)}
$ 
is  $\beta_2$.
     The base case then holds.
     Induction hypothesis:
     We assume that we have 
     $\sum_{i\in \mathcal{M}} \Delta \pi_i^j \1 \big(\Delta \pi_i\big)x_i = \alpha_j - x' \beta_j$, for all $j \leq k-1$.
     Induction step: The term
$
        \sum_{i\in \mathcal{M}} \Delta \pi_i^k \1 \big(\Delta \pi_i\big) x_i $
        is equal to
        $ \big (y^{(k)} - \sum_{j=1}^{k-1} f_{(j,k)}(\alpha_j - x' \beta_j) + (-1)^{k} \kappa^k x'\big)/f_{(k,k)},
        $ where the term
        $
        \big(y^{(k)} - \sum_{j=1}^{k-1} f_{(j,k)}\alpha_j \big)/f_{(k,k)} $ is ${\alpha_k} $, and the term $ \big( -\sum_{j=1}^{k-1} f_{(j,k)} \beta_j + (-1)^{k-1} \kappa^k \big)/f_{(k,k)}$ is equal to ${\beta_k}.
$
The induction step then holds true, and this completes the proof.
\subsection*{Proofs of \Cref{lem:case2} and \Cref{cor:identifiability-xbar-IM-R}} 
Without loss of generality, assume that during the time interval $(t_l, t_{l+1})$ the number of distinct positive payoff values among imitative groups is $n$. 
At first we introduce the following Lemma. 
 \begin{lemma}\label{lem:coefficient-of-beta-not-zero}
 The term   ${B}(t,t_l)$ is either zero over the time interval $(t_l, t_{l+1})$, or
   \begin{equation} \label{eq:proof_of_coefficient-of-beta-not-zero}
       \exists \hat{t}, \delta>0 \text{ s.t. }\forall t \in (\hat{t}, \hat{t} +\delta) \subset (t_l,t_{l+1})    {B}(t,t_l) \neq 0.
   \end{equation}  
 \end{lemma}
 Now, we proceed with the proof of \Cref{lem:case2}.
  If $t^* = 0$ or the vaccine supply is identically zero over $[0,t^*)$, we consider $[\epsilon, t_1)$ for $l=0$ and an arbitrarily small positive $\epsilon$.
For the sake of readability, 
here $n_l$ and $\bar{x}(t_l)$ are denoted by $n$ and $\bar{x}$, respectively.
Additionally,
the notation $c_i$ is used instead of the notation $\hat{\pi}_i(t_l)$ for $i \in \{1,2,\ldots,n\}$.
Then, in view of \Cref{lem:ci-gammai-based-on-alphas-betas}, the terms $\sum_{i=1}^n c_i^j\hat{x}_i$ for $j=0,\ldots,n$ can be recast as 
\begin{equation} \label{eq:V-im-br}
\scalebox{0.8}{$
\begin{pmatrix}
       1 & 1  & \cdots & 1 & 1  \\
       c_1 & c_2 & \cdots & c_n & 0\\
       \vdots & \vdots & \ddots & \vdots & \vdots \\
        c_1^{n} & c_2^{n}  & \cdots & c_n^n   & 0
    \end{pmatrix} 
    \begin{pmatrix}
        \hat{x}_1 \\
        \vdots \\
        \hat{x}_n \\
        \Bar{x}
    \end{pmatrix}
    = \bm \alpha(t) - x'(t) \bm \beta(t).
    $}
\end{equation}
The matrix on the left hand side of \eqref{eq:V-im-br}, is a Vandermonde matrix and is the same as \eqref{eq:vandermonde}.
 Along the same lines taken in the proof of \Cref{lem:case1}, we can obtain the following differential equation
\begin{equation}\label{eq:sigmaBar-R}
    \sigma_n\Bar{x} = 
    \sum_{l=0}^n (-1)^l \alpha_l(t) \sigma_{n-l}  - x'(t)  
    {B}(t,t_l).
\end{equation}
In view of \Cref{lem:coefficient-of-beta-not-zero}, two cases may happen:
\emph{Case 1.} there exists a non-measure zero interval contained in the subinterval $(t_l, t_{l+1})$ during which the  term ${B}(t,t_l)$ is not zero, or
\emph{Case 2.} the  term $  {B}(t,t_l)$ is zero over the time interval   $(t_l, t_{l+1})$.


\emph{Case 1.} given \eqref{eq:sigmaBar-R}, the value of $x'$ can be written in terms of the output and its derivatives and the parameters, i.e.,
\begin{equation} \label{eq:gammaR}
\scalebox{0.9}{$
x'(t) =  \frac{\sum_{l=0}^n (-1)^l \alpha_l(t) \sigma_{n-l} -\sigma_n \bar{x}}{ {B}(t,t_l)}.
$}
\end{equation}
In view of 
\eqref{eq:y-and-its-derivatives-R-and-M},
 we get
\begin{equation} \label{eq:proof_of_y_n1}
\scalebox{0.9}{$
    y^{(n+1)} (t)
    + (-1)^{n+1}\kappa^{n+1} x'(t) + \sum_{j=1}^{n+1} \big(f_{(j,n+1)}(t)\sum_{i=1}^n c_i^j\hat{x}_i(t) \big)  = 0.
    $}
\end{equation}
In view of \eqref{eq:V-im-br}, the variable $\hat{x}_i$ can be written in terms of $c_j$, $\alpha_j$, and $\beta_j$, for $j = 1,\ldots,n,$ as follows:
\begin{equation} \label{eq:gamma-im-br}
\scalebox{0.8}{$
    \begin{aligned}
        \hat{x}_i & (t)= 
         \frac{1}{\prod_{j \notin \{i,n+1\}} (c_i - c_j)} \Bigg( \sum_{k=1}^{n+1} c_i^{n-k} \bigg( \sum_{l=0}^{n} (-1)^{k+1+l}\alpha_l(t) \sigma_{k-1-l} \bigg)\Bigg) \\
        &  - \frac{\sum_{l=0}^n (-1)^l \alpha_l(t) \sigma_{n-l} -\sigma_n \bar{x}}{{B}(t,t_l) } \frac{1}{\prod_{j \notin \{i,n+1\}} (c_i - c_j)} 
         \times \Bigg( \sum_{k=1}^{n+1} c_i^{n-k} \bigg( \sum_{l=0}^{n} (-1)^{k+1+l}\beta_l(t) \sigma_{k-1-l} \bigg)\Bigg).
    \end{aligned}
    $}
\end{equation}
Let $g_i(t)$ denote the right hand side of \eqref{eq:gamma-im-br}.
By replacing $x'$ and 
$\hat{x}_i$ with \eqref{eq:gammaR} and \eqref{eq:gamma-im-br} in \eqref{eq:proof_of_y_n1},
we obtain the input-output relation
\begin{equation} \label{eq:dummy}
    y^{(n+1)}(t) + \sum_{j=1}^{n+1} \big(f_{(j,n+1)}(t)\sum_{i=1}^n c_i^jg_i(t) \big)
    + (-1)^{n+1}\kappa^{n+1} \frac{\sum_{l=0}^n (-1)^l \alpha_l(t) \sigma_{n-l} -\sigma_n \bar{x}}{{B}(t,t_l)} = 0.
\end{equation}
It is straightforward to show that  the highest degree of $y$ appearing in the denominator of the term $\beta_n$, denoted by $D(\beta_n)$ does not exceed $n + D(\beta_{n-1})$ where $D(\beta_1) = 1.$
Here, for the ease of notation we show $D(\beta_n)$ by $m$.
By multiplying  both sides of the above equation by $y^{m}{B}(t,t_l)$, for a large enough value of $m$, we get the input-output relation in the form of polynomial
\begin{equation}
\scalebox{0.9}{$
\begin{aligned} \label{eq:IO-4-R-M}
    y^{(n+1)}(t)y^{m}(t){B}(t,t_l)
    + 
   y^{m}(t){B}(t,t_l)
    \bigg(\sum_{j=1}^{n+1} \big(f_{(j,n+1)}(t)\sum_{i=1}^n c_i^jg_i(t) \big) \bigg)&
    \\ + (-1)^{n+1}\kappa^{n+1} y^m(t)\bigg(\sum_{l=0}^n (-1)^l \alpha_l(t) \sigma_{n-l} -\sigma_n \bar{x} \bigg)& = 0.
    \end{aligned}
    $}
\end{equation} 
It can be shown that neither
$\alpha_i$ nor $\beta_i$ for $i \leq n$ contain the term $y^{(n+1)}$, and
the term
$\sum_{i=1}^n c_i^jg_i $ does not contain it either.
In view of \Cref{lem:f-j-j+k-derivative-output}, the term $f_{(j,n+1)}$, $j \in \{1,\ldots,n+1\}$, does not contain $y^{(n+1)}$ either.
In the input-output equation, the only terms containing the $(n+1)$th order derivative of the output are 
$y^{(n+1)}(t) y^m(t){B}(t,t_l)$.
Hence, the coefficients of the monomials in the term $y^{(n+1)}(t) y^m(t){B}(t,t_l)$
are identifiable.
In view of \eqref{eq:betas},
this term reads as 
$
     \sum_{l=0}^n \big(\sigma_{n-l} \frac{\sum_{j=1}^{l-1} f_{(j,l)}y^{m-l}y^{(n+1)}\beta_j }{\kappa^l}+   (-1)^l\sigma_{n-l} y^{(n+1)}y^{m-l} \big).
$
The term $\sum_{j=1}^{l-1} f_{(j,l)}y^{m-l}\beta_j$ does not contain any monomials without the derivative of the output, see \Cref{lem:f-j-k-derivative-output}.
Hence,
the
 terms $\sigma_{n-l}$ for $l = 0, \ldots, n$, which are the coefficients  of the monomial $y^{(n+1)}y^{m-l}$ in the input-output equation, are identifiable.
If \emph{Case 2} happens, in view of \eqref{eq:sigmaBar-R}, the input-output relation will be 
$
    \sigma_n\Bar{x} - 
    \sum_{l=0}^n (-1)^l \alpha_l(t) \sigma_{n-l} = 0.
$
The identifiability of the terms $\sigma_{n-l}$, $l =0,\ldots, n$ can be shown similar to the steps taken in the proof of \Cref{lem:case1}.
 In view of the definition of $\sigma_k$, indeed by knowing the values of $\sigma_l$, $l = 1,2,\ldots, n$, we know the values of the elementary symmetric polynomials in the  distinct positive payoff values of the imitative groups for the time interval $(t_l, t_{l+1})$, which results in the identifiability of the set $\mathcal{P}_l$.
Now we move on to the proof of \Cref{cor:identifiability-xbar-IM-R}.
The input-output equation during the subinterval $\mathcal{L}$  equals to \eqref{eq:IO-4-R-M}.
For the identifiability of the parameter $\bar{x}$, we see that $\bar{x}$
appears only in the coefficient of $y^m$. 
We should look for other occurrences of the monomial $y^m$ in the generalized input-output equation  
\eqref{eq:IO-4-R-M}.
The terms in the first part of  
the equation \eqref{eq:IO-4-R-M} consist of monomials including $y^{(n+1)}$.
The term
$
y^m{B}(t,t_l)
    \bigg(\sum_{j=1}^{n} \big(f_{(j,n+1)}\sum_{i=1}^n c_i^jg_i \big) \bigg)
$
 consist of monomials with the (higher-order) derivatives of the output, due to $f_{(j,n+1)}$ where $j \leq n$.
The terms in
\begin{equation} \label{eq:fn+1-4-proof}
\scalebox{0.9}{$
y^m{B}(t,t_l)
    \bigg( \big(f_{(n+1,n+1)}\sum_{i=1}^n c_i^{n+1}g_i \big) \bigg)
    $}
\end{equation}
includes monomials without the derivatives of the output.
But in the following we show that the terms in 
\eqref{eq:fn+1-4-proof} which do not contain derivatives of the output
have a degree larger than $m$.
The term 
$
y^m{B}(t,t_l)
     c_i^{n+1}g_i
$
reads as
\begin{equation} 
\scalebox{0.8}{$
    \begin{aligned}
        &\frac{\sum_{l=0}^n (-1)^l y^m\beta_l \sigma_{n-l}}{\Pi_{j \notin \{i,n+1\}} (c_i - c_j)} \Bigg( \sum_{k=1}^{n+1} c_i^{n-k} \bigg( \sum_{l=0}^{n} (-1)^{k+1+l}\alpha_l \sigma_{k-1-l} \bigg)\Bigg) \\
        & \quad - \big(\sum_{l=0}^n (-1)^l \alpha_l \sigma_{n-l} -\sigma_n \bar{x}\big)\frac{1}{\Pi_{j \notin \{i,n+1\}} (c_i - c_j)}\\
        & \times \Bigg( \sum_{k=1}^{n+1} c_i^{n-k} \bigg( \sum_{l=0}^{n} (-1)^{k+1+l}\beta_ly^m \sigma_{k-1-l} \bigg)\Bigg).
    \end{aligned}$}
\end{equation}
From \eqref{eq:betas}, recall that the term $\beta_l$ equals
$
     y^m\beta_k = \frac{-\sum_{j=1}^{k-1} f_{(j,k)}\beta_jy^{m-k}}{(-1)^{k-1}\kappa^k} + y^{m-k}.
$
The first term in the right-hand side of the  equation consists of monomials with derivatives of the output.
Hence, the term $y^{m-k}$ would be the only term in  $y^m \beta_k$ which is without the derivative of the output.
In view of \Cref{cor:alphas-M-R-have-derivatives}, we know that only $\alpha_0$ does not contain the  derivatives of the output.
The product of the terms containing $y^m\beta_k$ and $\alpha_0$ will then result in monomials whose degrees with respect to $y$ is non-negative.
On the other hand, $f_{(n+1,n+1)}$ is equal to 
$(-1)^{n-1}\kappa^{n+1}y^{n+1}$.
As a result, 
 in equation \eqref{eq:fn+1-4-proof}, those monomials without the  derivatives  of the output have degree larger than $m$.
Hence,
the monomial $y^m$ in
 the equation \eqref{eq:IO-4-R-M} appears only as 
 $(-1)^{n+1}\kappa^{n+1}\sigma_n(1-\bar{x})y^m$.
Owing to the identifiability of $\kappa$ and $\sigma_n$, the identifiability of $\bar{x}$ is obtained. 
\subsection*{Proof of \Cref{lem:identifiability-kappa}}
If
$t^*>0$ and $u(t)$ is not identically zero during $[0,t^*)$.
Then, for $t \in [0,t^*)$ the system dynamics read as \eqref{eq:dotGammaBRandIM}. The derivative of the output then satisfies $\dot{y}(t) - \kappa u(t) = 0$.
For the time instants at which $u(t)$ is zero, so is the derivative of the output. 
Other than those time instants, we have $\kappa = \frac{\dot{y}(t)}{u(t)}$,
and the identifiability of the parameter $\kappa$ is then immediate.
If $t^*$ is zero, or the vaccine supply $u(t)$ is identically zero during $[0,t^*),$  there exists at least one non-refusing rationalists and two cases might happen:
\text{Case 1} There exists a time interval where some rationalists but no imitators have positive payoff values. 
In this case, as shown in the proof of \Cref{prop:unidentifiability-of-rationalist-payoff}, the input-output equation equals $\dot{y}(t) - \kappa (1 - y(t) - \bar{x})$, where $\bar{x}$ is the proportion of unvaccinated individuals with non-positive payoff values, and identifiability of the parameter $\kappa$ is immediate.
\text{Case 2}.  
    There exists a subinterval where  some rationalists and some imitators have positive payoff values, then, similar to the proof of \Cref{lem:case2}, it can be shown that in the input-output equation the term $y^{(n+1)}y^m$ is multiplied by $\sum_{l=0}^{n} \sigma_{n-l}/\kappa^l \big(\sum_{j=1}^{l-1} f_{(j,l)}\beta_j/y^l \big) + (-1)^k\sigma_{n-k}/y^k$. 
    In view of \eqref{eq:betas},
     the term $\beta_{n-1}$, 
     is a summation of some certain terms and $(1/y^{n-1})$.
     Here, the term $\beta_{n-1}$
      is multiplied by $f_{(n-1,n)} = \kappa^{n-1} a_n \dot{y} y^{n-2}$, where $a_n$ is some known number.
      It can be shown than
     the  monomial
    $y^{(n+1)}y^m \dot{y}\dot{y}/y^n+1$ has the coefficient of $\sigma a_n/\kappa$
     Hence, the coefficient $\sigma_0/\kappa$ is identifiable--note that $a_n$ is a known number, and in view of $\sigma_0 = 1$, parameter $\kappa$ is identifiable.
\subsection*{Proof of \Cref{lem:gi-xstar}}
For each time instant $t_l$, $l \in \{0,1, \ldots, {\s-1}\}$, construct a set consisting of the values of $\hat{x}_i(t_l)$, $i = 1,2,\ldots, n_l$, and denote it by $\mathcal{Q}_l.$
In view of \Cref{lem:case2}, the payoff values $\hat{\pi}_i(t_l)$ are identifiable.
With this and in view of the
second condition of the lemma,  the set $\mathcal{Q}_l$, including $\mathcal{Q}_0$, is known. 
In this case,  given $\dot{\hat{x}}_i(t) = -\kappa \hat{x}_i(t)y(t)\hat{\pi}(t_l)$ and the value of $\hat{x}_i(t_l)$,
the trajectory of $\hat{x}_i(t)$ for the time interval $t \in (t_l,t_{l+1})$ is  known.
Denote a set consisting of the values of $\hat{x}_i(t_l^-)$ for $l \in \{1,2,\ldots,\s\}$ and $i = 1,2,\ldots, n(t_{l-1})$ by $\mathcal{S}_l.$
Note that, in view of the continuity of the system states, the values of $x_i$, $i \in \mathcal{M}$, at $t_l^-$ and $t_l^+$ are equal. 
Now,
compare the members of the pair sets $\mathcal{Q}_{l}$ and $\mathcal{S}_l$ for $l = 1.$
The equal members of these pair sets are associated with the same imitative groups.
The condition $\vert \mathcal{S}_l \vert \neq \vert \mathcal{Q}_l \vert$ implies that the signs of the payoffs of some imitative groups change at $t_l$. More specifically, if there are some values in the set $\mathcal{S}_l$ ($\mathcal{Q}_l$) with no equivalent value in  the set $\mathcal{Q}_l$ ($\mathcal{S}_l$), the payoff values corresponding to those members changes to non-positive (positive) from positive (negative) values for the time interval $(t_{l},t_{l+1})$.
Note that thanks to the first condition of the lemma, each value of $\hat{x}_i$, for $i \in \{1,2,\ldots,n_l\}$, corresponds to exactly one imitative group.
By repeating the same procedure for the remaining intervals and matching the values of $\hat{x}_i$ at the time instants $t_l \in \{1,2,\ldots, \s-1\}$,
we may retrieve the trajectories associated with imitative groups having positive payoff values during at least one subinterval.
Accordingly, the corresponding positive payoff values are retrievable.
This implies that the set $\bm{\mathcal{N}}$ can be retrieved.
If $n_l \neq 0$, the proportion of unvaccinated rationalists with positive payoff values at time instant $t_l$  can be calculated using \eqref{eq:xprime-mainText}. 
Its trajectory during the time interval $(t_l,t_{l+1})$ then reads as $x'(t) = x'(t_l)\exp(-\kappa (t-t_l)).$ 
When $n_l = 0$, it is straightforward to obtain the input-output relation $\dot{y}(t) = \kappa (1-\bar{x}(t_l)) - \kappa y(t)$ which indicates the identifiability of $\bar{x}(t_l)$ and, in turn, $x'(t_l) = 1-y(t_l) - \bar{x}(t_l).$
Then the trajectory of $x'(t)$ during the time interval $(t_l,t_{l+1})$ reads as $x'(t) = x'(t_l)\exp(-\kappa (t-t_l)).$
Finally, for those time intervals where $\mathtt{r}(t_l) =0$, simply $x'(t)$ will be equal to $x'(t_l)$ during the time interval $t\in(t_l,t_{l+1}).$
\subsection*{Proofs of \Cref{lem:ij-entry-matrix-A} and \Cref{prop:identifiability-of-x-tstar}}
In view of \eqref{eq:dotGamma}, during the time interval $[t_{l-1}, t_{l})$, $l \in \{1,2,\ldots, \s\}$,
 the rate of change of the proportion of unvaccinated individuals in imitative group $i$ is
$
     \dot{x}_i(t) =-\kappa x_i(t) y(t)\Delta \pi_i(t_{l-1})\1\big(\Delta \pi_i(t_{l-1})\big)$, or equivalently,
 $
\frac{d{x_i}}{x_i} = -\kappa y(t) \Delta \pi_i(t_{l-1})\1\big(\Delta \pi_i(t_{l-1})\big) dt.
 $
 Taking integral of both sides results in 
 $
 \ln (x_i (t)/x_i(t_{l-1})) = -\kappa \Delta \pi_i(t_{l-1}) \1\big(\Delta \pi_i(t_{l-1})\big)\int_{t_{l-1}}^{t} y(\tau)d\tau,
 $
or equivalently,
 $\smash{
 x_i (t) = x_i(t_{l-1})\exp \big(\!-\kappa \Delta \pi_i(t_{l-1}) \1\big(\Delta \pi_i(t_{l-1})\big)\int_{t_{l-1}}^{t} y(\tau)d\tau \big).} 
$
 Substituting the value of $x_i(t_{l-1})$ for previous subintervals into the derived equation  repeatedly results
 \begin{equation} \label{eq:gamma_t0-imitators}
 \scalebox{0.79}{$
     x_i (t) \!= \!x_i(t^*)\exp \!\big(\!-\kappa \!\sum_{j=1}^{j=l-1} \!\big(\Delta \pi_i(t_{j-1}) \1\big(\Delta \pi_i(t_{j-1})\big)\int_{t_{j-1}}^{t_{j}} y(\tau)d\tau \big)\! -\!\kappa
     \big(\Delta \pi_i(t_{l-1}) \1\big(\Delta\!\pi_i(t_{l-1})\big)\!\int_{t_{l-1}}^{t}\!y(\tau)d\tau\!\big)
     \big)$}
 \end{equation}
 for  $i \in \mathcal{M}$ and  $t \in (t_{l-1},t_l)$.
For a rationalist group $i$ and time interval $t \in (t_{l-1},t_l)$, we have
 $x_i(t) = x_i(t^*)\exp \Big( -\kappa \sum_{j=1}^{l-1} (t_j - t_{j-1})\1\big(\Delta \pi_i(t_{j-1})\big) - \kappa(t - t_{l-1})\1\big(\Delta \pi_i(t_{l-1})\big)  \Big).
$
Plugging this equation and \eqref{eq:gamma_t0-imitators}  into $
    y(t) =   1 -  \sum_{i} x_i(t) 
$ results in 
\begin{equation*} \scalebox{0.75}{$
    \begin{aligned}
     y(t)& =1 -  \sum_{i \in \mathcal{R}} x_i(t^*)\exp \Big(\!-\kappa \!\sum_{j=1}^{l-1} (t_j - t_{j-1})\1\big(\Delta \pi_i(t_{j-1})\big)\!
      - \kappa(t - t_{l-1})\1\big(\Delta \pi_i(t_{l-1})\big) \! \Big)  \\
      & - \sum_{i \in \mathcal{M} }  x_i(t^*)\exp \! \Big(-\kappa \sum_{j=1}^{j=l-1} \Delta \pi_i(t_{j-1}) \1\big(\Delta \pi_i(t_{j-1})\big)\int_{t_{j-1}}^{t_{j}} y(\tau)d\tau  -\kappa
     \Delta \pi_i(t_{l-1}) \1\big(\Delta \pi_i(t_{l-1})\big)\int_{t_{l-1}}^{t} y(\tau)d\tau 
     \Big), 
      \end{aligned}
      $}
 \end{equation*}
 for $t \in (t_{l-1},t_l).$
Hence, we may build the matrix $\matr{A}$ as in equation
\eqref{eq:MatrixA-R-M}
with $\n$ columns and $m$ rows, each row corresponds to some time instant $\Tilde{t}_j$ in $[t^*,T).$
By defining
$\x^* = \big(x_1(t^*),\ldots, x_{\n}(t^*) \big)^\top$ and an $m$-dimensional vector $\bm b$ with
$b_j = 1-y(\Tilde{t}_j)$
as its $j^\text{th}$ component, we have
$
        \matr{A} \x^* = \bm b.
$
Hence, the $(i,k)$ entry of matrix $\matr{A}$ is $x_k(\Tilde{t}_i)/x_k(t^*)$.
The maximum rate of vaccination $\kappa$ is identifiable. 
Under the conditions of the proposition, that is, the signs of the payoff values of all $\n$ groups and
 the ordering among the positive payoff values of imitative groups are known for  
 all $t \leq \tilde{t}_m$,
 the identifiability of the
payoff functions of the imitative groups are  satisfied.
Hence,
the entries of matrix $\matr{A}$ are known, and
the  identifiability of $\x(t^*)$ is guaranteed if  matrix $\matr{A}$ is full rank. 
\subsection*{Proof of  \Cref{prop:identifiability-x-star-2}}
The proofs of these results are similar to those of \Cref{lem:gi-xstar} and \Cref{prop:identifiability-of-x-tstar}, respectively, and are omitted accordingly.
\subsection*{Proof of \Cref{prop:identifiability-of-x0}}
Let
$\bm x_1(t)$ and 
$\bm x_2(t)$ be
two solution trajectories to \eqref{eq:gammaVector} starting from two different initial conditions, i.e., $\bm x_1(0) \neq \bm x_2(0).$
In view of  \Cref{prop:identifiability-of-x-tstar} or \Cref{prop:identifiability-x-star-2}, the vector $\bm x(t^*)$ is identifiable and hence unique.
Accordingly, we have $\bm x_1(t^*) = \bm x_2(t^*).$
Therefore, two trajectories must cross at some time $\hat{t}^* \in (0,t^*]$.
Set the time instant $\hat{t}^*$ as a new initial time. 
Two cases may happen based on the value of $\hat{t}^*$.
\textit{Case 1)} $\hat{t}^* \neq t^*$.
As at least one rationalist group has positive payoff value,
the local Lipschitzness of the vector field $\bm \h$ can be  shown.
With this and in view of the continuity of the input $u(t)$,  according to
main theorem of uniqueness, the solutions to the different initial conditions $\bm x^1(0)$ and $\bm x^2(0)$ should be the same at $\hat{t}^*  - \varepsilon$ as well, where $\varepsilon$ is a small enough positive value.
However, this contradicts our assumption that the two trajectories $\bm x^1(t)$ and $\bm x^2(t)$ merge at $\hat{t}^*$.
As a result, the two solutions starting from these two different initial conditions do not merge during the time interval $[0,t^*]$.
Therefore, we should have different values of $\bm x(t^*)$ for different initial conditions $\bm x(0)$.
This contradicts our knowledge that  $\bm x (t^*)$ is the same regardless of the initial condition.
Hence, the two initial conditions should be the same, and consequently $\bm x(0)$ is identifiable.
\textit{Case 2)}
 $\hat{t}^* = t^*$. 
In this case, consider $\Tilde{\bm h}(t, \bm x)$ which is the same as 
$\bm h$, but is continuous on $\hat{t}^* $.
Main theorem of uniqueness implies that the two solutions to the vector field $\Tilde{\bm h}$ correspond to the initial conditions $\bm x^1(0)$ and $\bm x^2(0)$ should be the same for the time interval $(\hat{t}^* - \epsilon, \hat{t}^* + \epsilon)$.
On the other hand, 
 the solutions of both $\bm h$ and $\Tilde{\bm h}$ are the same during $(\hat{t}^*  - \epsilon, \hat{t}^*]$.
A similar reasoning as in \textit{Case 1} can then be applied, leading to a contradiction.
\subsection*{Proof of \Cref{thm:the-final-theorem}}
Given \Cref{prop:identifiability-of-x0}, the proof of the theorem is immediate and, accordingly, is omitted.

\subsection*{Proof of \Cref{lem:unidentifiability-of-negative-payoffs}}
Without loss of generality, assume that the $l^\text{th}$ component of the vector payoff of the $i^\text{th}$ group $\bm \Delta \pi_i$, is not positive and
denote its value by $c <0$.
With abuse of notation by $t_0$, we mean $t^*$.
Then, 
the rate of change of $x_i$ over the time interval $(t_{l-1},t_{l})$ reads as $\dot{x}_i(t) = -\kappa x_i(t) \1(c)y(t,c)c$ if $i$ is an imitative group, and $\dot{x}_i(t) = -\kappa x_i(t) \1(c)$ if $i$ is a rationalist group, where $y(t,c)$ is the output when the $l^\text{th}$ component of the vector $\bm \Delta \pi_i$ is $c$.
In either case, $\dot{x}_i(t)$ equals zero for $t \in (t_{l-1},t_{l})$ which yields $x_i(t) = x_i(t_{l-1})$.
Let $\epsilon$ be an arbitrarily positive number.
Let $c' = c - \epsilon$.
If the value of $c$ is replaced with $c'$, the rate of change of $x_i$ still remains zero during the time interval $(t_{l-1},t_l)$.
Consequently, the value of $y(t,c')$ which is the output when the $l^\text{th}$ component of the vector $\bm \Delta \pi_i$ is replaced with $c'$ is equal to $y(t,c)$ for all $t \geq 0.$
This implies the unidentifiability of the parameter $c$.
\subsection*{Proof of \Cref{prop:unidentifiability-of-rationalist-payoff}}
The input-output equation during the time interval $(T_l,T_{l+1})$, $l \neq 0$, where only some rationalists have positive payoff values, equals $\dot{y}(t) = \kappa (1-y(t) - \bar{x})$, where $\bar{x}$ is the proportion of unvaccinated individuals with non-positive payoff values during the time interval $[T_l,T_{l+1})$.
As the payoff values do not appear in the input-output equation, they are not identifiable.
For the time interval $(0,t^*)$, the input-output equation reads as $\dot{y}(t) - \kappa u(t) = 0$ and, consequently, the payoff values are not identifiable.
The input-output equation for the time interval $[t^*,T_1
)$ is the same as $(T_l,T_{l+1})$, $l \neq 0$, and similar result is obtained.
\subsection*{Proof of \Cref{lem:induction-For-Dot-y}}
For the sake of readability, by $n$ we mean $n_l$, and the notation $c_i$ is used instead of $\hat{\pi}_i(t_l)$ for $i \in \{1,\ldots,n_l\}$.
    We prove by induction.
    
    Base case: from $\dot{y}(t) 
    = \kappa y(t) \sum_{i=1}^{n_l}\hat \pi_i(t_l)\hat{x}_i(t),$ we have 
    $\dot{y} = \kappa y \sum_{i=1}^n c_i\hat{x}_i$ and hence by defining $f_{(1,1)} = \kappa y$, the base case holds.
    Induction step: Assume that for $k-1 \in  \mathbb{Z}_{>0}$,  we have 
     $
    y^{(k-1)} = \sum_{j=1}^{k-1} \big( f_{(j,{k-1})} \sum_{i=1}^n  c_i^j \hat{x}_i \big).
    $
    By taking the derivative of both sides of the above equation and using the relation
$
    \frac{d(\sum_{i=1}^n c_i^k\hat{x}_i)}{dt} = -\kappa y \sum_{i=1}^n c_i^{k+1}\hat{x}_i,
$
    we have
$    y^{(k)} = \sum_{j=1}^{k-1} \frac{d\big({f}_{(j,k-1)}\big)}{dt} \sum_{i=1}^n  c_i^j \hat{x}_i + \sum_{j=1}^{k-1} 
    f_{(j,k-1)}\frac{d(\sum_{i=1}^n  c_i^j \hat{x}_i)}{dt}
 $
 and, in turn,
 $y^{(k)} =
    \sum_{j=1}^{k-1} \dot{f}_{(j,k-1)} \sum_{i=1}^n  c_i^j \hat{x}_i - 
    \sum_{j=1}^{k-1}
    \kappa yf_{(j,k-1)} \sum_{i=1}^n  c_i^{j+1}\hat{x}_i.
    $
    The term $y^{(k)}$ can then be rewritten  as 
 $
   y^{(k)}   = \sum_{j=1}^{k} \bigg(f_{(j,k)} \sum_{i=1}^n  c_i^{j}\hat{x}_i \bigg),
$
    where $f_{(0,\cdot)} = 0$.
    The induction step holds, and this completes the proof.
\subsection*{Proof of \Cref{cor:f-k-k}}
\begin{proof}
In view of \eqref{eq:f-j-k},
\begin{equation*}
    \begin{aligned}
    f_{(k,k)} &= \dot{f}_{(k,k-1)} - \kappa y f_{(k-1,k-1)}, 
    \end{aligned}
\end{equation*}
and by definition
${f}_{(k,k-1)}$ is equal to zero, implying $\dot{f}_{(k,k-1)}$ is zero.
Hence, 
$f_{(k,k)} = -\kappa f_{(k-1,k-1)}$ which is a recursive equation with initial condition $f_{(1,1)} = \kappa y,$ which makes the result trivial.
\end{proof}
\subsection*{Proof of \Cref{lem:f-j-k-derivative-output}}
\begin{proof}
    We prove by induction.
    Base case: We show that $f_{(j,j+1)}$
    does not have any monomials without the  derivatives of the output.
    We also prove this by induction, where the base case is to show that $f_{(1,2)}$ does not have any monomials without the (higher-order) derivatives of the output.
    This is straightforward, as $f_{(1,2)} = \kappa \dot{y}$.
    The induction hypothesis would be that $f_{(j-1,j)}$ does not contain any monomials without the  derivatives of the output.
    The induction step is then to show that the term
    $f_{(j,j+1)}$ does not contain any monomials without the (higher-order) derivatives of the output.
    The term $f_{(j,j+1)}$ equals $\dot{f}_{(j,j)} - f_{(j-1,j)}$, where  $\dot{f}_{(j,j)}$ equals $(-1)^{j+1} \kappa^jj y^{j-1}\dot{y}$, and the term $f_{(j-1,j)}$ does not include any monomials without the derivatives of the output based on the induction hypothesis.
    Hence, the difference between these two terms also enjoys this property.
    Hence, the induction step of this induction, and, in turn, the base case of the main induction hold true.

    Induction hypothesis: Assume  that for $k>1$, the term $f_{(j,j+k-1)}$
    does not have any monomials without the derivatives of the output.

    Induction Step: We  show that $f_{(j,j+k)}$
    does not have any monomials without the derivatives of the output.
    We show that this holds by using another induction.
    The base case is to show that the induction step holds true for $f_{(1,k+1)}$.
    The term $f_{(1,k+1)}$ equals $\dot{f}_{(1,k)} - \kappa yf_{(0,k)}$. 
    The term $f_{(0,k)}$ is zero. With this and after some mathematical manipulation, we obtain     
$ f_{(1,k)} = \kappa y^{(k-1)}$ and, in turn, 
    $\dot{f}_{(1,k)}$ does not contain any monomials without the  derivatives of the output.
The base case thus holds.
The induction hypothesis would be that $f_{(j-1,j+k-1)}$ does not contain any monomials without the  derivatives of the output.
The induction step is then to show this for $f_{(j,j+k)}$.
The term equals $f_{(j,j+k)} = \dot{f}_{(j,j+k-1)} - \kappa y f_{(j-1,j+k-1)}$.
Based on the induction hypothesis, the second term in the right-hand side of the equation does not contain any monomials without the  derivatives of the output.
The induction hypothesis of the main induction implies that the term ${f}_{(j,j+k-1)}$ for $k>1$ does not contain any monomials without the  derivatives of the output nor does it its derivative. 
This proves the induction step and hence the proof is complete.
\end{proof}
\begin{corollary} \label{cor:betas-M-R-have-derivatives}
    For $k > 0$, the term $y^n\beta_k$, where $\beta_k$ is defined in \eqref{eq:betas}, does not contain  the monomial $y^n$.
\end{corollary}
\begin{proof}
The term $y^n\beta_k$ reads as 
\begin{equation}
    y^n\beta_k = y^{n-k}\frac{(-1)^k}{\kappa^k}\sum_{j=1}^{k-1} f_{(j,k)}\beta_j +y^{n-k}.
\end{equation}
Based on \Cref{lem:f-j-k-derivative-output}, for $j<k$ the term $f_{(j,k)}$, and, in turn, the first term on the right hand side of the above equation does not contain any monomials without the derivatives of the output. This makes the result trivial.
\end{proof}
\subsection*{Proof of \Cref{lem:f-j-j+k-derivative-output}}
\begin{proof}
We prove by induction.

Base case: We  show that $f_{(j, j+1)}$ does not contain $y^{(2)}$ or higher.
$f_{(1,2)}$ equals $ \dot{f}_{(1,1)} - \kappa y f_{(0,k)}$.
By plugging
$\dot{f}_{(1,1)} =\kappa \dot{y}$ and $f_{(0,k)}$ into $f_{(1,2)}$, we have 
$f_{(1,2)} = \kappa \dot{y}$, which does not include $y^{(2)}$ or higher-order derivatives of the output.
Induction hypothesis: It is assumed that
$f_{(j-1,j)}$ does not include $y^{(2)}$ or higher-order derivatives of the output.
Induction step: 
$f_{(j,j+1)}$
equals
$\dot{f}_{(j,j)} - \kappa y f_{(j-1,j)}$
The derivative of ${f}_{(j,j)}$ equals $j(-1)^{j+1} \kappa^j \dot{y} y^{j-1}$ and does not contain $y^{(2)}$ or higher order.
Based on the induction step, the term 
$f_{(j-1,j)}$ does not include $y^{(2)}$ or higher order. With this  $f_{(j,j+1)}$  does not include $y^{(2)}$ or higher order derivatives either. This proves the induction step and hence the base case of our main induction is proven. 

Induction Hypothesis: Assume that the term $f_{(j,j+k-1)}$ does not contain $y^{(k)}$ or higher order.
Induction Step: We should show that $f_{(j,j+k)}$ does not contain $y^{(k+1)}$ or higher order. We have
$
f_{(1,1+k)} = \dot{f}_{(1,k)} - \kappa y f_{(0,k)}.
$
The last second term in the right hand side equals zero, and
based on the induction step, the first term in the right hand-side of the equation does not contain $y^{(k+1)}$ or higher order, nor does $f_{(1,1+k)}$.
Similar to the base case, using induction it can be easily shown that
$f_{(j,j+k)}$ does not include contain $y^{(k+1)}$ or higher orders.
Hence, the induction step holds, and this completes the proof.
\end{proof}
\subsection*{Proof of \Cref{lem:coefficient-of-beta-not-zero}}
 \begin{proof} 
     If the term ${B}(t,t_l)$  is not zero over the entire time interval $(t_l, t_{l+1})$, then there exists some time instant $\hat{t} \in (t_l, t_{l+1})$ at which $  {B}(\hat{t},t_l) \neq 0$. 
     Assume that $  {B}(\hat{t},t_l) = \epsilon$ for an arbitrarily positive value $\epsilon$, a similar reasoning holds true for a negative value $\epsilon$.
       Let 
     $\Tilde{t} = \inf \{ t > \hat{t} \vert  {B}(t,t_l) =0 \}$.
     If $\Tilde{t}$ does not exists or  $\Tilde{t} \geq t_{l+1}$, then \eqref{eq:proof_of_coefficient-of-beta-not-zero} is satisfied for $\delta = t_{l+1} - \hat{t}$.
     Otherwise, 
     in view of \eqref{eq:betas},  the term ${B}(t,t_l)$ is rational and its denominator consists of some order of $y$.
     The value of $y$ is strictly greater than zero  for $t > t^*$, where $t^*$ is defined in \Cref{ass:two-subintervals-shortage-surplus}, and, in turn, for $t \in (t_l, t_{l+1})$. 
     The nominator of function ${B}(t,t_l)$ is a differential polynomial in $y$ and its higher-order derivative.
     The output $y$ and its higher order derivatives are continuously differentiable over the time interval $(t_l, t_{l+1})$.
     Thus ${B}(t,t_l)$ is continuously differentiable over the time interval $(t_l, t_{l+1})$, and, in turn, is Lipschitz continuous over any compact subset of $(t_l, t_{l+1})$, i.e.,
     $\vert {B}(\hat{t},t_l) - {B}(\Tilde{t},t_l) \vert \leq L \vert \hat{t} - \Tilde{t} \vert $ for some positive value $L$.
    Hence, we obtain
$ \vert \epsilon -  {B}(\Tilde{t},t_l) \vert \leq L \vert \hat{t} - \Tilde{t} \vert 
            $ which results in  $\epsilon  \leq L (\Tilde{t} - \hat{t})$ and subsequently $\epsilon / L + \hat{t} \leq \Tilde{t}.
$
    As a result, for each $t$ in the interval $(\hat{t}, \hat{t} + L/\epsilon)$, the term $ {B}(t,t_l)$ remains positive.
 \end{proof}
 \subsection*{Proofs of \Cref{lem:t-existence-imitation} and \Cref{prop:t-discontinuous-imitation}} 
\begin{proof}
For the ease of notation and without loss of generality assume that the first $\m$ groups are imitative and their payoff values change at time instant $\hat{t}$.
The proportion of unvaccinated individuals in each group $i$, $x_i$, for $i=1,2,\ldots, \p$, is continuous for all $t \in [0,T)$ and, in turn,
    the output $y(t)$ satisfies $y(\hat{t}^-) = y(\hat{t}^+)$.
        We prove this lemma by contradiction, that is, we assume that  at time instant $\hat{t}$,
        although $\text{det }(\matr{P}) \neq 0$,
        the output is $(\m+1)$-times differentiable.
        The necessary condition for the differentiability of  the output of order $(\m+1)$, is the existence and the equality of
        $y^{(k)}(\hat{t}^+)$ and $ y^{(k)}(\hat{t}^-)$  for all $k\leq \m$.
         Denote the  value of 
    $\Delta \pi_i(t)$  at $t = \hat{t}^+$ by
    $\bar{c}_i$
    the value of  
    $\Delta \pi_i(t)$ at $t = \hat{t}^-$ by $c_i$.
    For the second derivative of ${y}$ to exist at $\hat{t}$, the condition $\dot{y}(\hat{t}^-) = \dot{y}(\hat{t}^+)$ must hold.
    In view of $\dot{y}(t) = \kappa y \sum_{i=1}^n c_ix_i(t)\1(c_i) + \kappa x'(t)$
where $n = \vert \mathcal{M} \vert$ is the total number of imitative groups and $x'(t)$ is the proportion of unvaccinated rationalist individuals with positive payoff values,
    the condition
    holds only if $ \sum_{i=1}^n c_i x_i(\hat{t})\1(c_i) = \sum_{i=1}^n \bar{c}_i x_i(\hat{t})\1(\bar{c}_i)$.
    A similar reasoning can be applied for the existence of the higher-order derivatives of the output.
    In summary, if the output is $(\m+1)$-times differentiable, then
    $y^{(k)}(\hat{t}^+) = y^{(k)}(\hat{t}^-)$  must  hold for $k = 1, \ldots, \m$, yielding
    \begin{equation}  \label{eq:proof-t-identifiable-eq1}
        \begin{aligned}
            \sum_{i=1}^n c_i \1(c_i) x_i &= \sum_{i=1}^n \bar{c}_i \1(\bar{c}_i) x_i,\\
              \sum_{i=1}^n c_i^2 \1(c_i) x_i &= \sum_{i=1}^n \bar{c}_i^2 \1(\bar{c}_i) x_i,\\
             & \vdots \\
               \sum_{i=1}^n c_i^\m \1(c_i) x_i &= 
              \sum_{i=1}^n \bar{c}_i^\m \1(\bar{c}_i) x_i,
        \end{aligned}
    \end{equation}
     For a non-positive value $z$, we have $z \1(z) = 0$ and, accordingly,
    the equation \eqref{eq:proof-t-identifiable-eq1} can be rewritten as
     \begin{equation} \label{eq:proof-t-identifiable-eq2}
        \begin{aligned}
        \sum_{i=1}^n \big(c_i \1(c_i) - \bar{c}_i\1(\bar{c}_i)\big)x_i &= 0,\\
         \sum_{i=1}^n \big(c_i^2 \1(c_i) - \bar{c}_i^2\1(\bar{c}_i)\big)x_i &= 0,\\
             & \vdots \\
         \sum_{i=1}^n \big(c_i^\m \1(c_i) - \bar{c}_i^\m\1(\bar{c}_i)\big)x_i &= 0.
        \end{aligned}
    \end{equation}
    For those groups whose payoff values remain unchanged at $\hat{t}$, we 
    have $c_i = \bar{c}_i$. This yields $ \sum_{i=\m+1}^n (c_i^k - \bar{c}_i^k)x_i=0$, for $k>0$.
\Cref{eq:proof-t-identifiable-eq2} equals to
\begin{equation} \label{eq:proof-t-identifiable-eq3}
  \begin{pmatrix}
       1 & 1  & \cdots & 1   \\
       p^{[1]}_1 & p^{[2]}_1 & \cdots & p^{[\m]}_1\\
       \vdots & \vdots & \vdots & \vdots   \\
        p^{[1]}_{\m-1} & p^{[2]}_{\m-1} & \cdots & p^{[\m]}_{\m-1}
    \end{pmatrix}
    \begin{pmatrix}
        \big(c_1 \1(c_1) - \bar{c}_1\1(\bar{c}_1)\big)x_1(\hat{t}) \\
         \big(c_2 \1(c_2) - \bar{c}_2\1(\bar{c}_2)\big)x_2(\hat{t}) \\
         \vdots \\
         \big(c_\m \1(c_\m) - \bar{c}_\m\1(\bar{c}_\m)\big)x_\m(\hat{t})
    \end{pmatrix} 
    =
    \begin{pmatrix}
        0 \\
        0 \\
        \vdots \\
        0
    \end{pmatrix}.
\end{equation}
Since we assumed $\text{det }(\matr{P}) \neq 0$, 
the vector 
$$
   \begin{pmatrix}
        \big(c_1 \1(c_1) - \bar{c}_1\1(\bar{c}_1)\big)x_1(\hat{t}) \\
         \big(c_2 \1(c_2) - \bar{c}_2\1(\bar{c}_2)\big)x_2(\hat{t}) \\
         \vdots \\
         \big(c_\m \1(c_\m) - \bar{c}_\m\1(\bar{c}_\m)\big)x_\m(\hat{t})
    \end{pmatrix} 
$$
must be equal to zero, implying that
 either $c_i\1(c_i)$ equals  $\bar{c}_i\1(\bar{c}_i)$ or $x_i$ is zero.
The value of $x_i$, for $i = 1,2,\ldots, n$, cannot be zero over the finite time interval $[0,T)$ and the former case implies that $c_i\1(c_i) = \bar{c}_i\1(\bar{c}_i)$ for all $i=1,2,\ldots, \m$, which contradicts our assumption.
Therefore, should matrix $\matr{P}$ be of full rank, \eqref{eq:proof-t-identifiable-eq3} and consequently \eqref{eq:proof-t-identifiable-eq2} do not hold and, in turn, the output is not differentiable up to order $(\m+1)$.
We reach to a contradiction and this completes the proof of \Cref{lem:t-existence-imitation}.
As for the proof of \Cref{prop:t-discontinuous-imitation}, in view of \Cref{lem:t-existence-imitation}, as long as $\text{det }(\matr{P}) \neq 0$, the output is at most $\m$-differentiable at time instant $\hat{t}$ where the payoff values of $\m$ groups of imitators change.
Hence, by 
excluding the payoff values for which $\text{det }(\matr{P}) = 0$, a time instant cannot be an endpoint if the output is  $(\vert \mathcal{M} +1 \vert)$-times differentiable at that point.
\end{proof}
\subsection{Identifiability theorem}
\begin{theorem}[\cite{eisenberg2013identifiability}] \label{thm:Eisenberg}
   The parameters of a rational function ODE model 
   \begin{equation} 
    \Sigma^{\bm \theta} = 
    \left\{
    \begin{aligned}
        & \dot{\bm x}(t) = \bm f\big({t},\bm x(t), \bm u(t), \bm \theta \big), \\
        & \bm y(t, \bm \theta) = \bm g \big(\bm x(t), \bm \theta \big),
    \end{aligned}
    \right.
\end{equation}
    where $\bm x \in \mathbb{R}^{\mathtt{x}}$,
    $\bm \theta \in \bm \Theta \subset \mathbb{R}^{\theta}$,
    $\bm u \in \bm U \subset \mathbb{R}^{\mathtt{u}}$, $\bm f: \mathbb{R} \times \mathbb{R}^{\mathtt{x}} \times \mathbb{R}^{\mathtt{u}} \times  \bm \Theta  \to \mathbb{R}^{\mathtt{x}}$, $\bm g:  \mathbb{R}^{\mathtt{x}} \times  \bm \Theta  \to \mathbb{R}^{\mathtt{y}}$, 
   are globally 
structurally identifiable if and only if the map  from the
parameters to the coefficients of a set of input–output equations is
injective,
regardless of how the input–output equations are generated.
\end{theorem}
\section{Further elaboration on \Cref{example:first}} \label{sec:further}
\textbf{Identifiability analysis}. 
According to \Cref{lem:identifiability-kappa}, the maximum rate of vaccination $\kappa$ is identifiable.
Recall that if  the ordering of the payoff values of imitators during each subinterval is known, in view of \Cref{cor:identifiability-of-the-cis-case2}, their identifiability is guaranteed.
Matrix $\matr{A}$ then reads as
\begin{equation} \label{eq:example-A-first}
\begin{blockarray}{ccccccc}
& G_1 & G_2 & G_3 & G_4 & G_5 & G_6 \\
    \begin{block}{c(cccccc)}
    \tilde{t}_1 &  1& 1 & 1& 1 & 1 & 1 \\
   \tilde{t}_2 &  \exp(-\kappa(\tilde{t}_2 - t^*))   & \exp(-\kappa(\tilde{t}_2 - t^*)) &a_{23} & a_{24} & 1 & 1 \\
   \tilde{t}_3 &   \exp(-\kappa(\tilde{t}_3 - t^*))   &  \exp(-\kappa(\tilde{t}_3 - t^*)) &a_{33} & a_{34} & 1 & 1 \\
   \tilde{t}_4 &   \exp(-\kappa(\tilde{t}_4 - t^*))   & \exp(-\kappa(\tilde{t}_4 - t^*)) & a_{43}& a_{44} & 1 & 1 \\
    \tilde{t}_5 &   \exp(-\kappa( \tilde{t}_5 - t^*))   & \exp(-\kappa( \tilde{t}_5 - t^*)) & a_{53}& a_{54} & 1 & 1 \\
  \tilde{t}_6 &   \exp(-\kappa(\tilde{t}_6 - t^*))   & \exp(-\kappa(\tilde{t}_6 - t^*)) & a_{63}& a_{64} & 1 & 1 \\
    \end{block}
    \end{blockarray},
\end{equation}
 where $\tilde{t}_1 = 0$, 
 $\tilde{t}_2 \in (t^*,t_1)$, 
  $\tilde{t}_3 = t_1$,
$\tilde{t}_4,\tilde{t}_5 \in (t_1,T)$, $\tilde{t}_6 = T$, 
 $$ a_{23} =   \exp \big(-\kappa \Delta \pi_s(t^*) \int_{t^*}^{\tilde{t}_2} y(\tau) d\tau  \big), \quad  a_{33} =   \exp \big( -\kappa \Delta \pi_s(t^*) \int_{t^*}^{\tilde{t}_3} y(\tau) d\tau  \big),$$
  $$ a_{24} =   \exp \big(-\kappa \Delta \pi_n(t^*) \int_{t^*}^{\tilde{t}_2} y(\tau) d\tau  \big), \quad  a_{34} =   \exp \big( -\kappa \Delta \pi_n(t^*) \int_{t^*}^{\tilde{t}_3} y(\tau) d\tau  \big),$$
where $\Delta \pi_s$ (resp. $\Delta \pi_n$) is the perceived payoff gain of imitative seniors (resp. non-seniors) and for $j = 4,5,6$ we have
   $$ a_{j3} =   \exp \big(-\kappa \Delta \pi_s(t^*) \int_{t^*}^{t_1} y(\tau) d\tau   -\kappa \Delta \pi_s(t_1) \int_{t_1}^{\tilde{t}_j} y(\tau) d\tau  \big),$$
     $$ a_{j4} =   \exp \big(-\kappa \Delta \pi_n(t^*) \int_{t^*}^{t_1} y(\tau) d\tau   -\kappa \Delta \pi_n(t_1) \int_{t_1}^{\tilde{t}_j} y(\tau) d\tau  \big).$$
The first two columns and the last two columns of matrix $\matr{A}$ are  the same, and, accordingly,  matrix $\matr{A}$ is not full rank.
 The conditions in \Cref{prop:identifiability-of-x-tstar}  are  not then satisfied (\Cref{rem:multiple-refusers-rationalists}).
We then consider the total proportion of rationalists, $\rho_1 + \rho_2$, and total proportion of vaccine refusers, $\rho_5 + \rho_6$,  as two unknown parameters and delete two arbitrarily rows of the matrix \eqref{eq:example-A-first}, say the  fourth and fifth ones.
This results in
\begin{equation} \label{eq:example-second-A} 
\begin{blockarray}{ccccc}
& G_1+G_2 & G_3 & G_4  & G_5 + G_6 \\
    \begin{block}{c(cccc)}
    \tilde{t}_1 &  1& 1 & 1 & 1 \\
     \tilde{t}_2  & \exp(-\kappa(\tilde{t}_2 - t^*)) &a_{23} & a_{24} & 1\\
    \tilde{t}_3 &   \exp(-\kappa( \tilde{t}_3 - t^*)) &a_{33} & a_{34}& 1 \\
     \tilde{t}_6 &   \exp(-\kappa( \tilde{t}_6 - t^*))  & a_{63}& a_{64} & 1 \\
    \end{block}
    \end{blockarray},
\end{equation}
which is  full-rank.
Consequently, the proportions of imitative seniors $\rho_3$, imitative non-seniors $\rho_4$,  total rationalists $\rho_1 + \rho_2$, and vaccine refusers $\rho_5 + \rho_6$ are globally identifiable in view of \Cref{prop:identifiability-of-x-tstar}.

\Cref{prop:identifiability-of-x-tstar} requires the information on the ordering of the imitative payoff values, though.
Thankfully, it can be shown that, in this example, the conditions in \Cref{lem:gi-xstar}
are satisfied, indicating that the set $\bm{\mathcal{N}}$ is identifiable and, in turn, the evolution of proportions of unvaccinated imitators and rationalists  during the vaccination program can be retrieved without knowledge of payoff value ordering.
In addition, the evolution of the imitative positive payoff values is also retrievable.
However, the proportion of each imitative group is identifiable if the function from the set $\bm{\mathcal{N}}$ to the set $\mathcal{M}$ is known. 
In this case, the proportions  $\rho_3$ and $\rho_4$ can be calculated.

\textbf{Parameter estimation.} We simulated the model with true parameter values and recorded the output, i.e., the cumulative proportion of vaccinated individuals $y(t)$, and its derivatives up to the third order at six time instants--four time instants in the first interval $[0,20)$ and two time instants in the second interval $[20,40).$
We also recorded the integral of the output at two time instants $t_1 = 20$ and $T =40$.
We then obtained the input-output equation  and estimated the maximum rate of vaccination $\kappa$, the proportion  of vaccine-refusers $\rho_5 + \rho_6$, and the set $\mathcal{P}(0)$  using the four data points in the first time interval and the nonlinear least square method.
In the next step, we used the values of the output and its derivatives at the two data points in the second interval and estimated the set $\mathcal{P}(20)$.
The proportions of the groups were estimated using two approaches.

In the first approach, we assumed that the ordering of the payoff values was known and, in turn, besides the sets $\mathcal{P}(0)$ and $\mathcal{P}(20)$, the values of $\Delta \pi_s(0)$, $\Delta \pi_n(0)$, $\Delta \pi_s(20)$, and $\Delta \pi_n(20)$ were identifiable.
Then, we formed a system of linear equations $\matr{A}' \bm x^* = \bm b$, where matrix $\matr{A}'$ was similar to \eqref{eq:example-second-A} just that the second row and the last column were deleted. 
The vector $\bm b$ was equal to $\big(1-y(t^*) - \rho_5 - \rho_6, 1-y(t_1) - \rho_5 - \rho_6, 1-y(T) - \rho_5 - \rho_6 \big)^\top$ and $\bm x^* = (\rho_1 + \rho_2, \rho_3, \rho_4)^\top.$
By solving the equation $\bm x^* = \text{inv}(\matr{A}')\bm b$, the proportions of $\rho_1 + \rho_2$, $\rho_3$, and $\rho_4$ were obtained.
The estimated values are reported in the third row of the \Cref{tab:example2}.

In the second approach, we did not use the information about the  payoff value ordering.
Using the equations  \eqref{eq:xprime-mainText} and \eqref{eq:gamma-im-br-main-text}, the evolution of the unvaccinated rationalists  and the evolution of unvaccinated imitative groups were obtained. 
Accordingly, the population proportion of rationalists was estimated from \eqref{eq:xprime-mainText}, evaluated at $t = 0.005$. 
Similarly, those of imitators were estimated from evaluating \eqref{eq:gamma-im-br-main-text} at $t = 0.005$ (\Cref{fig:enter-label}).
Note that in this case, the obtained imitative proportions cannot generally be associated to specific groups, i.e., seniors or non-seniors.
Moreover, we set the instant of evaluation to $0.005$, to avoid from division by zero, as the output appears in the denominators of equations  \eqref{eq:xprime-mainText} and \eqref{eq:gamma-im-br-main-text} and  is zero at time zero. 
The estimated  proportions were $0.3908$ for rationalists and $0.15$ and $0.2499$ for the two imitative groups. 

\begin{figure}
    \centering
    \includegraphics[width=0.95\linewidth]{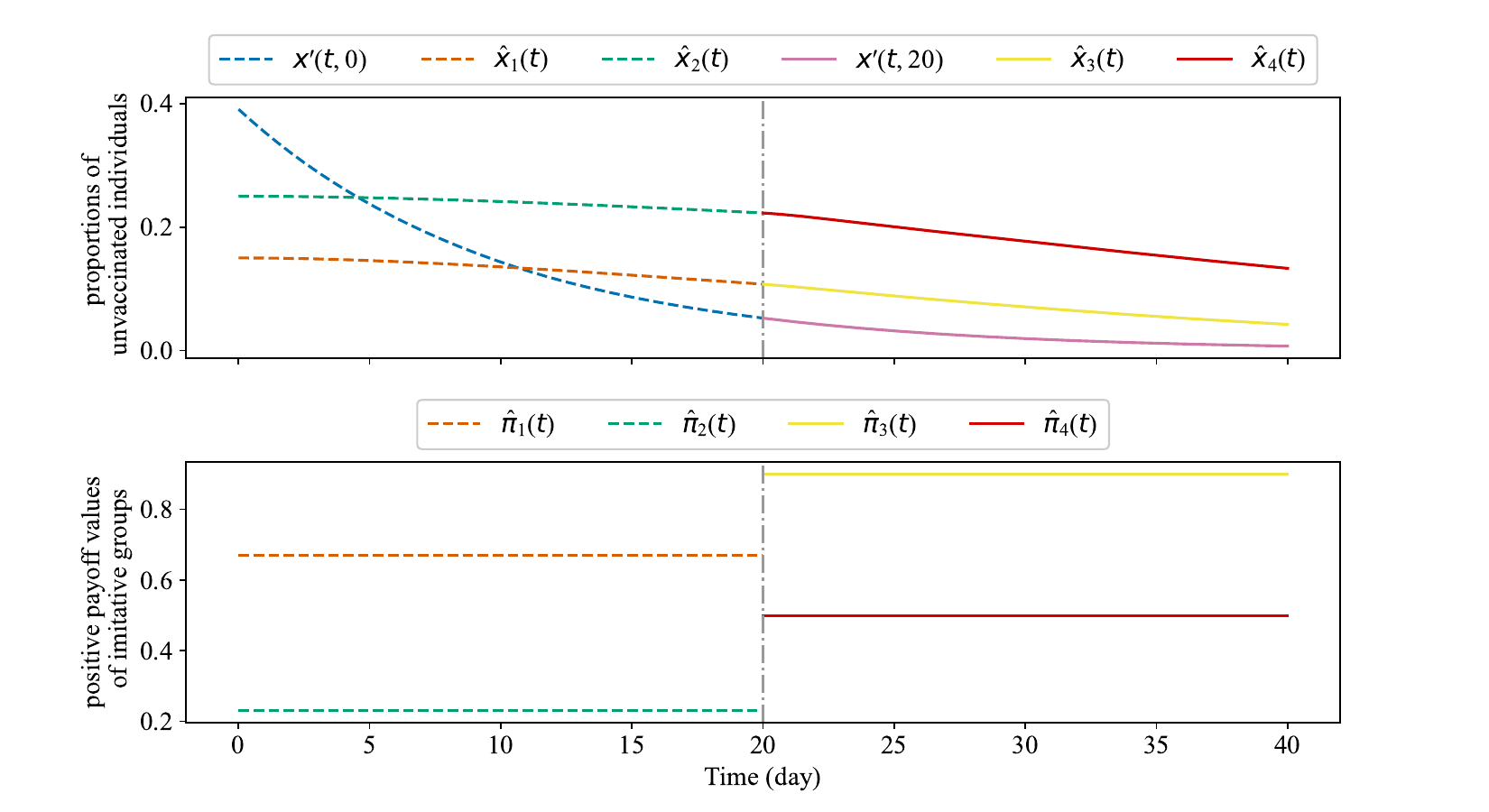}
    \caption{Retrieving the evolution of the proportion of unvaccinated rationalists, as well as the proportions and payoffs of unvaccinated imitators during the vaccination program $[0,40)$ in \Cref{example:first}, without knowledge of payoff value ordering.
   Upper panel:
   over the time interval $[0,20)$ (resp. $[20,40)$), the evolution of the proportion of unvaccinated rationalist $x'(t,0)$ (resp. $x'(t,20)$) was obtained using the equation \eqref{eq:xprime-mainText}.
    In the absence of information about the ordering of the  payoff values, for time interval $[0,20)$ (resp. $[20,40)$), using equation \eqref{eq:gamma-im-br-main-text}, we obtained two graphs ($\hat{x}_1(t), \hat{x}_2(t)$) (resp.
    $\hat{x}_3(t), \hat{x}_4(t)$) each one  corresponding to one imitative group.
    Under the conditions of \Cref{lem:gi-xstar}, the pieces can be concatenated, and for each imitative group, the whole trajectory can be retrieved, i.e., the pieces $\hat{x}_1(t)$ and  $\hat{x}_3(t)$ (resp. $\hat{x}_2(t)$ and  $\hat{x}_4(t)$) are associated with the same imitative group.
    Lower panel:
    the payoff values $\hat{\pi}_1$ and $\hat{\pi}_3$ (resp. $\hat{\pi}_2$ and $\hat{\pi}_4$) are associated with the same imitative group.
     }
    \label{fig:enter-label}
\end{figure}
\begin{table}  
\caption{The estimated parameter values in \Cref{example:first}, where no information on the ordering of the two imitative groups' payoff values was fed to the optimizer. 
 The second (resp. third) row reports the means and the coefficients of variation ($100\% \times $standard deviation divided by the mean), denoted by CV, of the parameters obtained from fitting the model to  $1000$ simulated  data, which were made noisy with Gaussian noise having a mean of zero and a standard deviation equal to ten percent of the mean (resp. Poisson noise structure). 
The initial parameter guesses were drawn from a uniform distribution within $50\%$ of the true values.
}\label{tab:example2-noordering}
\centering 
\resizebox{\textwidth}{!}{
\begin{tabular}{|l|r|r|r|r|r|r|r|r|r|r|} 
\hline
  &  $\kappa$ & $\rho_1+ \rho_2$ & $\rho_3$ & $\rho_4$ & $\rho_5 + \rho_6$ & $\Delta \pi_s(0)$ & $\Delta \pi_n(0)$  & $\Delta \pi_s(t_1)$ & $\Delta \pi_n(t_1)$  \\ 
  \hline
         \footnotesize  \makecell[l]{Mean (Gaussian) \\ CV }   & \makecell{$0.106$\\$19\%$}& \makecell[r]{$0.383$\\$16\%$} & \makecell[r]{$0.159$\\$33\%$} & \makecell[r]{$0.262$\\$23\%$} & \makecell[r]{$0.195$\\$30\%$} & \makecell[r]{$0.64$\\$32\%$} & \makecell[r]{$0.225$\\$49\%$} & \makecell[r]{$0.785$\\$20\%$} & \makecell[r]{$0.498$\\$27\%$} 
         \\
    \hline
         \footnotesize \makecell[l]{Mean (Poisson) \\ CV}   & \makecell{$0.106$\\$18\%$}& \makecell[r]{$0.382$\\$16\%$} & \makecell[r]{$0.155$\\$34\%$} & \makecell[r]{$0.26$\\$19\%$} & \makecell[r]{$0.203$\\$25\%$} & \makecell[r]{$0.63$\\$33\%$} & \makecell[r]{$0.24$\\$39\%$} & \makecell[r]{$0.8$\\$17\%$} & \makecell[r]{$0.51$\\$24\%$} \\
         \hline
\end{tabular}
}
\end{table}
\begin{figure}
    \centering
\includegraphics[width=0.95\linewidth]{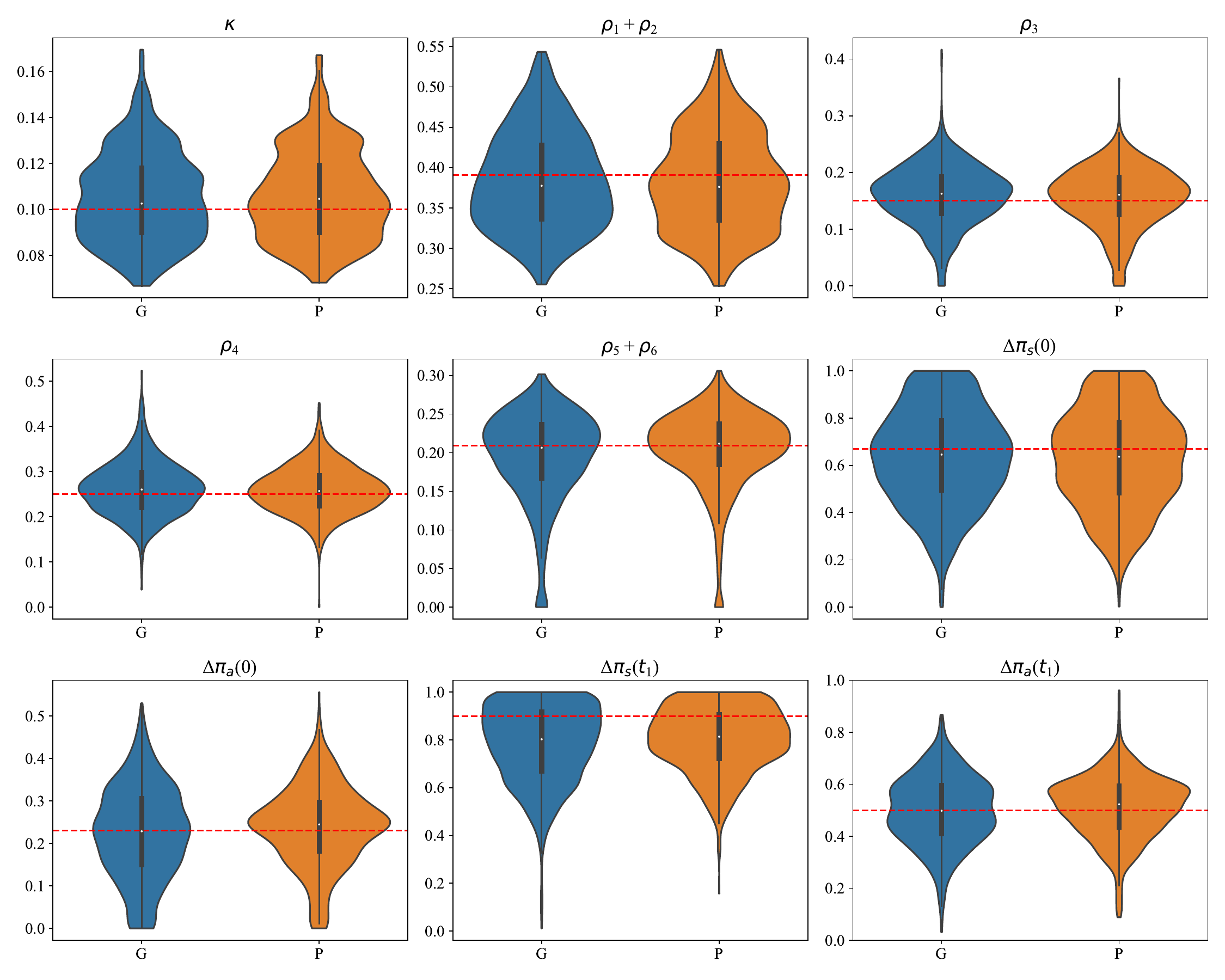}
    \caption{The violin plots of the estimated parameter values obtained from
    fitting the model to  $1000$ simulated noisy data.
    The letter ``G'' (resp. ``P'') stands for Gaussian (resp. Poisson) noise structure.
    The dashed lines show the true parameter values.
    No  ordering  among the payoff values was assumed.
    The distributions of the estimated values using the two noise structures are quite similar  for parameters maximum rate of vaccination $\kappa$ and  the proportions of imitative $\rho_3$ and $\rho_4$.
     }
    \label{fig:violin-no-ordering}
\end{figure}
\textbf{Practical Identifiability. }\Cref{fig:violin-no-ordering} represents the violin plots of the estimated parameter values in \Cref{example:first}, where no ordering among the payoff values are assumed during parameter estimation.
The huge differences between the violin plots of the estimated values for $\Delta \pi_s(0)$ (and $\Delta \pi_s(t_1)$) in 
\Cref{fig:violin} and \Cref{fig:violin-no-ordering} stem from the considered bounds on the parameter values.
When the ordering of the payoff values was known (\Cref{fig:violin}), rather than the value of $\Delta \pi_s$, the  difference between the two parameters $\Delta \pi_s$ and $\Delta \pi_n$ was estimated.
The difference was forced to be non-negative.
For the case where the ordering was not forced (\Cref{fig:violin-no-ordering}), the values of  parameters  $\Delta \pi_n$ and $\Delta \pi_s$ were directly estimated with bounds $[0,1]$.
 \section*{Acknowledgments}
This research was enabled in part by support provided by the Digital Research Alliance of Canada (alliancecan.ca).
\bibliographystyle{siamplain}
\bibliography{references}
\end{document}

%% file: share.tex
\usepackage{lipsum}
\usepackage{amsfonts}
\usepackage{graphicx}
\usepackage{epstopdf}
\usepackage{algorithmic}
\usepackage{mathrsfs}
\usepackage{calligra} 
\usepackage{xcolor,soul,framed} 
\usepackage{bm}
\usepackage[font=small,labelfont=bf]{caption}
\colorlet{shadecolor}{yellow}
\usepackage{import}
\usepackage{array}
\usepackage{dsfont}
\usepackage{xcolor}
\usepackage{tabularx}
\usepackage{makecell}
\usepackage{tikz}
\usepackage{blkarray}
\usepackage{amsmath,amsfonts,amssymb}
\usepackage{mathtools}
\usepackage{cuted}
\usepackage{cite}
\usepackage{cleveref}
\usepackage{amssymb}
\usepackage{cases}
\usepackage{empheq}
\newcommand{\matr}[1]{\mathbf{#1}} 
\newcommand{\Z}{\mathds{Z}}

\newcommand{\x}{\bm x}
\newcommand{\s}{\mathtt{s}}

\newcommand{\h}{\bm h}

\newcommand{\1}{\mathtt{1}}

\newcommand{\n}{\mathtt{n}}
\newcommand{\m}{\mathtt{m}}
\newcommand{\p}{\mathsf{p}}

\ifpdf
  \DeclareGraphicsExtensions{.eps,.pdf,.png,.jpg}
\else
  \DeclareGraphicsExtensions{.eps}
\fi

\newsiamremark{remark}{Remark}
\newsiamremark{example}{Example}
\crefname{hypothesis}{Hypothesis}{Hypotheses}
\newsiamthm{claim}{Claim}
\newsiamthm{assumption}{Assumption}
\newsiamremark{cor}{Corollary}
\headers{Identifying Vaccination Dynamics}{A. Aghaeeyan, M. A. Lewis, and P. Ramazi}

\title{Identifiability analysis of vaccination decision-making dynamics\thanks{Submitted to the editors DATE.
\funding{
This work was partially supported by
 Natural Science and Engineering Research Council of Canada (NSERC) Discovery Grants.
MAL gratefully acknowledges the Gilbert and Betty Kennedy Chair in Mathematical Biology.}}}

\author{Azadeh Aghaeeyan\thanks{Department of Mathematics and Statistics, Brock University, ON, Canada
  (\email{aaghaeeyan@brocku.ca}).}
\and Mark A. Lewis\thanks{Department of Mathematics and Statistics and Department of Biology,
University of Victoria, BC, Canada
(\email{mlewis@uvic.ca}).}
\and Pouria Ramazi\thanks{Department of Mathematics and Statistics, Brock University, ON, Canada
  (\email{pramazi@brocku.ca}).}
}

\usepackage{amsopn}

\makeatletter
\newcommand*{\addFileDependency}[1]{
  \typeout{(#1)}
  \@addtofilelist{#1}
  \IfFileExists{#1}{}{\typeout{No file #1.}}
}
\makeatother

